\documentclass[letterpaper,11pt]{article}

\usepackage{pdfpages} 
\usepackage[margin=1in]{geometry}
\usepackage{amsfonts}
\usepackage{amsmath}
\usepackage{amssymb}
\usepackage{amsthm}
\usepackage{thmtools}
\usepackage{xfrac}
\usepackage{nicefrac}
\usepackage{bm}
\usepackage{bbm}
\usepackage{mathtools}
\usepackage{mdframed}
\usepackage[inline]{enumitem}
\global\mdfdefinestyle{myframe}{leftmargin=.75in,rightmargin=.75in,linecolor=black,linewidth=1.5pt,innertopmargin=10pt,innerbottommargin=10pt}

\usepackage{zref-savepos}

\newcounter{fillparbox}

\usepackage{tikz}
\usetikzlibrary{calc}

\usepackage{graphicx}
\graphicspath{{graphics/}}

\usepackage{xspace}
\usepackage{mleftright}
\usepackage{forloop}
\definecolor{darkgreen}{rgb}{0,0.5,0}
\usepackage{hyperref}
\hypersetup{
    unicode=false,            
    colorlinks=true,          
    linkcolor=blue!70!black,  
    citecolor=darkgreen,      
    filecolor=magenta,        
    urlcolor=blue!70!black    
}

\usepackage{algorithm}
\usepackage[noend]{algpseudocode}

\usepackage{fix-cm}

\newfloat{protocol}{htbp}{loa}
\floatname{protocol}{Protocol}
\makeatletter\newcommand\l@protocol{\@dottedtocline{1}{1.5em}{2.3em}}\makeatother

\usepackage[immediate]{silence}
\usepackage{sectsty}
\DeactivateWarningFilters[temp]

\allsectionsfont{\boldmath}

\usepackage{anyfontsize}
 
\newtheorem{theorem}{Theorem}[section]	
 \newtheorem{lemma}[theorem]{Lemma}
 \newtheorem{definition}[theorem]{Definition}
 
 \newtheorem{corollary}[theorem]{Corollary}

 \newtheorem{observation}[theorem]{Observation}

 \newtheorem{proposition}[theorem]{Proposition}

\usepackage[capitalize, nameinlink, nosort]{cleveref}
\crefname{theorem}{Theorem}{Theorems}
\Crefname{lemma}{Lemma}{Lemmas}
\Crefname{claim}{Claim}{Claims}
\Crefname{fact}{Fact}{Facts}
\Crefname{remark}{Remark}{Remarks}
\Crefname{observation}{Observation}{Observations}
\Crefname{line}{Line}{Lines}
\Crefname{protocol}{Protocol}{Protocols}
\crefalias{AlgoLine}{line}
\crefname{property}{Property}{Properties}

\makeatletter
\newcounter{algorithmicH}
\let\oldalgorithmic\algorithmic
\renewcommand{\algorithmic}{%
  \stepcounter{algorithmicH}
  \oldalgorithmic}
\renewcommand{\theHALG@line}{ALG@line.\thealgorithmicH.\arabic{ALG@line}}
\makeatother

\usepackage{todonotes}

\newcommand{\e}{\varepsilon}
\newcommand{\eps}{\e}

\newcommand{\defcal}[1]{\expandafter\newcommand\csname c#1\endcsname{{\mathcal{#1}}}}
\newcommand{\defbb}[1]{\expandafter\newcommand\csname b#1\endcsname{{\mathbb{#1}}}}
\newcounter{calBbCounter}
\forLoop{1}{26}{calBbCounter}{
    \edef\letter{\Alph{calBbCounter}}
    \expandafter\defcal\letter
		\expandafter\defbb\letter
}

\newcommand{\SMkM}{{\texorpdfstring{\texttt{SMkM}}{SMkM}}\xspace}
\newcommand{\FESMkM}{\textnormal{{\textsc{First Element} {\SMkM}}}}
\newcommand{\nnR}{{\bR_{\geq 0}}}
\DeclareMathOperator{\rank}{rank}

\newcommand{\FEA}{{\texttt{FEAlg}}}
\algnewcommand\myand{\textbf{and} }
\algnewcommand\myor{\textbf{or} }
\newcommand{\chainPn}{CHAIN\textsubscript{$p$}$(n)$\xspace}
\newcommand{\chainPm}{CHAIN\textsubscript{$p$}$(m)$\xspace}
\newcommand{\chainPpn}{CHAIN\textsubscript{$p+1$}$(n)$\xspace}

\newcommand{\chain}[2]{{\ensuremath{\text{CHAIN}_{#1}(#2)}}}
\newcommand{\ALG}{{\text{ALG}}\xspace}
\newcommand{\PRT}{{\text{PRT}}\xspace}

\newcommand{\MSMBM}{{\texttt{MSM\-Bipartite\-Matching}}}
\DeclareMathOperator{\poly}{poly}

{\makeatletter
 \gdef\xxxmark{%
   \expandafter\ifx\csname @mpargs\endcsname\relax 
     \expandafter\ifx\csname @captype\endcsname\relax 
       \marginpar{xxx}
     \else
       xxx 
     \fi
   \else
     xxx 
   \fi}
 \gdef\xxx{\@ifnextchar[\xxx@lab\xxx@nolab}
 \long\gdef\xxx@lab[#1]#2{{\bf [\xxxmark #2 ---{\sc #1}]}}
 \long\gdef\xxx@nolab#1{{\bf [\xxxmark #1]}}
 \long\gdef\xxx@lab[#1]#2{}\long\gdef\xxx@nolab#1{}%
}

\title{Submodular Maximization Subject\\to Matroid Intersection on the Fly}
\author{Moran Feldman\thanks{Email: \href{mailto:moranfe@cs.haifa.ac.il}{moranfe@cs.haifa.ac.il}. Research supported in part by the Israel Science Foundation (ISF) grants no. 1357/16 and 459/20.}\\University of Haifa \and Ashkan Norouzi-Fard\thanks{Email: \href{mailto:ashkannorouzi@google.com}{ashkannorouzi@google.com}.}\\Google Research \and Ola Svensson\thanks{Email: \href{mailto:ola.svensson@epfl.ch}{ola.svensson@epfl.ch}. Research supported by the Swiss National Science Foundation project 200021-184656 ``Randomness in Problem Instances and Randomized Algorithms.''} \\EPFL  \and Rico Zenklusen\thanks{Email: \href{mailto:ricoz@ethz.ch}{ricoz@ethz.ch}. Research supported in part by Swiss National Science Foundation grant 200021\_184622. This project has received funding from the European Research Council (ERC) under the European Union's Horizon 2020 research and innovation programme (grant agreement No 817750).} \\ ETH Zurich}
\date{}

\begin{document}

\maketitle              

\begin{abstract}
Despite a surge of interest in submodular maximization in the data stream model, there remain significant gaps in our knowledge about what can be achieved in this setting, especially when dealing with multiple constraints.
In this work, we nearly close several basic gaps in submodular maximization subject to $k$ matroid constraints in the data stream model.
We present a new hardness result showing that super polynomial memory in $k$ is needed to obtain an $o(\sfrac{k}{\log k})$-approximation.
This implies near optimality of prior algorithms.
For the same setting, we show that one can nevertheless obtain a constant-factor approximation by maintaining a set of elements whose size is independent of the stream size.
Finally, for bipartite matching constraints, a well-known special case of matroid intersection, we present a new technique to obtain hardness bounds that are significantly stronger than those obtained with prior approaches. 
Prior results left it open whether a $2$-approximation may exist in this setting, and only a complexity-theoretic hardness of $1.91$ was known.
We prove an unconditional hardness of $2.69$.

\end{abstract}

\newtoggle{InIntroduction} \toggletrue{InIntroduction}
\section{Introduction} \label{sec:introduction}

A set function $f\colon 2^\cN \rightarrow \mathbb{R}$ over a ground set $\cN$ is \emph{submodular}  if 
 \begin{align*}
   \label{eq:submodular}
   f(u\mid A) \geq f(u\mid B)  \qquad \mbox{for all $A \subseteq B \subseteq \cN$ and $u\in \cN \setminus B$}
	\enspace,
 \end{align*}
 where, for a subset $S\subseteq \cN$ and an element $u\in \cN$, we denote by $f(u\mid S) := f(S \cup \{u\}) - f(S)$ the marginal contribution of $u$ with respect to $S$. 
We say that $f$ is \emph{monotone} if $f(A) \leq f(B)$ for any $A\subseteq B \subseteq \cN$.

 The definition of submodular functions captures the natural property of diminishing returns, and the study of these functions has a rich history in optimization with numerous applications (see, e.g., the book~\cite{schrijver2003combinatorial}). 
Already in 1978, Nemhauser, Wolsey, and Fisher showed that a natural greedy algorithm achieves a tight approximation guarantee of  $\tfrac{e}{e - 1}$ for selecting the most valuable subset $S\subseteq \cN$ of cardinality at most $\rho$ (see~\cite{nemhauser1978analysis} for the algorithm's analysis and~\cite{feige1998threshold,nemhauser1978best} for matching hardness results).  Since then, significant work has been devoted to  extending  their result to more general constraints.

A natural generalization of a cardinality constraint is the class of matroid constraints.
While matroid constraints are much more expressive than cardinality constraints, both constraints often enjoy the same (or similar) algorithmic guarantees.
Indeed, for the problem of maximizing a monotone submodular  function subject to a single matroid constraint, C{\u{a}}linescu, Chekuri, P{\'{a}}l, and Vondr{\'{a}}k~\cite{calinescu2011maximizing} developed the continuous greedy method, and showed that it extends the $\frac{e}{e-1}$-approximation guarantee to this more general setting. Moreover, for the maximization of a monotone submodular  function subject to $k \geq 2$ matroid constraints, there is a $(k + 1)$-approximation guarantee by Fisher, Nemhauser and, Wolsey~\cite{fisher1978analysis}, which was improved to $k + \eps$ by Lee, Sviridenko, and Vondr{\'{a}}k~\cite{lee2010submodular} when the number $k$ of matroid constraints is considered to be a constant.

While these algorithms are efficient in the traditional sense, i.e., they run in polynomial time in the offline ``RAM'' model, recent applications in data science and machine learning~\cite{SubmodularWWW} with very large-scale problem instances have motivated the need for very space-efficient algorithms. In particular, it is interesting to study algorithms whose memory footprint is \emph{independent of the ground set size}.\footnote{Technically, a logarithmic dependence on the ground set size is unavoidable because, at the very least, the algorithm has to store the indices of the elements in its solution. However, we wish to have a space complexity whose dependence on the ground set size is limited to this unavoidable logarithmic dependence.} The task of designing such algorithms for (monotone) submodular function maximization has become a very active research area, especially in the context of the popular data stream computational model. Recent progress has resulted in a tight understanding of data stream algorithms\footnote{Data stream algorithms are sometimes called ``streaming algorithms''; however, in this paper we reserve the term ``streaming algorithms'' to data streaming algorithms whose space complexity is poly-logarithmic in the natural parameters of the problem.} for maximizing monotone submodular functions with a single cardinality constraint: one can obtain a $2$-approximation for this problem using a simple ``threshold''-based algorithm that requires only $\tilde{O}(\rho)$ memory~\cite{badanidiyuru2014streaming,conficml0001MZLK19}, where $\rho$ is the maximum number of elements allowed in the solution, and this is essentially optimal unless one is willing to have a space complexity that is linear in the size of the ground set~\cite{feldman_2020_one-way}.
However, our understanding of data stream algorithms for more general constraint families is currently much more limited.
Closing this gap for natural settings of  multiple matroid constraints is the motivation for our work. 
%
%
%
%
%
%
%

Formally, we term the problem that we study \texttt{Submodular Maximization subject to $k$ Mat\-roid Const\-raints} ({\SMkM}). In this problem, we are given $k \geq 2$ matroids $M_1 = (\cN, \cI_1), M_2 = (\cN, \cI_2), \dotsc,\allowbreak M_k = (\cN, \cI_k)$ sharing a common ground set $\cN$, and a non-negative submodular function $f\colon 2^\cN \to \nnR$. Our goal is to find a \emph{common independent set} $S \subseteq \cN$ (i.e., $S$ is independent in all the matroids) that maximizes $f(S)$. In the data stream version of this problem, the elements of the ground set $\cN$ appear one by one on a stream, and the algorithm should make a single pass over this stream and construct its output set $S$. (Some papers allow also algorithms that can do a few sequential passes over the stream; however, we consider the more practical and fundamental single-pass setting.)

The above high-level description of {\SMkM} hides some important technical details regarding the way in which the objective function $f$ and the matroid constraints are accessed. In the literature about matroids, it is customary to assume that the access of an algorithm to a matroid is done via an \emph{independence oracle} that, given a set $S \subseteq \cN$, indicates whether $S$ is independent in the matroid. This notion can be extended to the intersection of $k$ matroids in two natural ways:
\begin{enumerate*}[label=(\roman*)]
\item having a single \emph{common independence oracle}, which indicates whether $S$ is independent in this intersection (i.e., in all the matroids), or
\item having $k$ independence oracles, one per matroid.
\end{enumerate*}
Let us first consider the model with weaker access to the matroids, i.e., the one with a common independence oracle. This model already allows the implementation of a simple algorithm that greedily adds to its solution every element that does not violate feasibility. This natural greedy algorithm is a $k$-approximation for the special case of {\SMkM} in which $f$ is simply the cardinality function (i.e., $f(S) = |S|$)~\cite{jenkyns1978efficacy,korte1978analysis} using a space complexity of $\tilde{O}(\rho)$, where $\rho$ is the \emph{common rank} of the $k$ matroids, i.e., the cardinality of a maximum cardinality common independent set in the $k$ matroids of the \SMkM instance.\footnote{For general {\SMkM}, the state-of-the-art algorithm with a space complexity of $\tilde{O}(\rho)$ obtains a slightly worse approximation ratio of $O(k \log k)$ with a common independence oracle~\cite{haba2020streaming}. Moreover, this algorithm is applicable even to the more general class of $k$-extendible constraints.}
We can show that this simple algorithm is almost best possible when access to the matroid is restricted to calls to a common independence oracle, unless one is willing to have a space complexity that is linear in $n \coloneqq |\cN|$.
\begin{theorem} \label{thm:multiple_matroids_single_oracle_hardness}
A data stream algorithm for {\SMkM}, whose only access to the matroids is via the common independence oracle, and with expected approximation ratio $k - \eps$ (for some $\eps \in [0, k - 1)$), must use $\Omega(\sfrac{\eps n}{k^5 \log k})$ memory. This holds even when the task is to find a maximum size common independent set in $k$ partition matroids, and the common rank of these matroids is $k$.
%
\end{theorem}
The proof of \cref{thm:multiple_matroids_single_oracle_hardness} is based on carefully defining $k$ matroids such that stream prefixes lead to restricted matroids with many indistinguishable elements. This allows for hiding a large optimal solution. See \cref{sec:inapproximability_oracle} for details.

Given this inapproximability result, we turn our focus to the model in which we have access to a separate independence oracle for every matroid. A $4k$-approximation algorithm with space complexity $\tilde{O}(\rho)$ was given for {\SMkM} in this model by Chakrabarti and Kale~\cite{chakrabarti2015submodular} when the objective function $f$ is guaranteed to be monotone, and $O(k)$-approximation algorithms with similar space complexities were later obtained for the general case by Chekuri et al.~\cite{chekuri2015streaming} and Feldman et al.~\cite{feldman2018do}. Our first main result shows that improving over the approximation guarantees of these algorithms by more than a logarithmic factor requires super polynomial space in $k$. Specifically, we prove the following theorem.
\begin{restatable}{theorem}{thmHardnessMultipleStreaming}
    Any data stream algorithm for {\SMkM} that finds an $\alpha$-approximate solution with probability at least $2/3$ uses memory at least $\Omega\left(e^{k/(8\alpha)}/{k^2}\right)$ assuming $\alpha \leq k / (32 \ln k) $. This holds even when the task is to find a maximum size common independent set in $k$ partition matroids, and the common rank of these matroids is $O(\alpha) = O(k / \log k)$.
    \label{thm:hardness_multiple_streaming}
\end{restatable}

The technique to prove \cref{thm:hardness_multiple_streaming} is discussed in \cref{sec:hardness_intersection}. Interestingly, this technique also implies that any (even preemptive) online algorithm for {\SMkM} must also have an approximation ratio of $\Omega(\sfrac{k}{\log k})$. 
We remark that this lower bound is asymptotically the same as the well-known approximation hardness for $k$-dimensional matching~\cite{hazan_2006_complexity}, which is a special case of the intersection of $k$ matroids.
However, note that this hardness does not carry over to our setting as our model has no restriction on computational power but only on memory.

In light of \cref{thm:hardness_multiple_streaming}, it is arguably surprising that one can get essentially a $2$-approximation for {\SMkM} with space complexity independent of $n$.\footnote{For simplicity, the space complexity stated in \cref{th:matroid_intersection_alg_short} assumes that every element of the ground set can be stored in $O(1)$ space. Without this assumption, we get the unavoidable logarithmic dependence of the space complexity on $n$. Similarly, we also make the standard assumption that the value of $f(S)$ can be stored in constant space for every set $S \subseteq \cN$.}

\begin{restatable}{theorem}{thMatroidIntersectionAlgShort} \label{th:matroid_intersection_alg_short}
\iftoggle{InIntroduction}{For every $\eps \in (0, \sfrac{1}{7})$, there exists a $(2 + O(\eps))$-approximation data stream algorithm for {\SMkM} with space complexity $\poly(k^{\rho^2k}, \eps^{-\rho})$.}{For every $\eps \in (0, \sfrac{1}{7})$, there exists a $(2 + O(\eps))$-approximation data stream algorithm for {\SMkM} with space complexity $\tilde{O}((\frac{k^{2\rho k}}{\eps^3})^{\rho + 1} \cdot (\log \rho + \log k + \log \eps^{-1})^{\rho + 2} \rho)$. If the objective function $f$ of {\SMkM} is guaranteed to be monotone, the space complexity of the algorithm can be improved to $\tilde{O}((\frac{k^{\rho k}}{\eps})^{\rho + 1} \cdot (\log \rho + \log k + \log \eps^{-1})^{\rho + 2} \rho)$.}
\end{restatable}

For monotone objective functions, \cref{th:matroid_intersection_alg_short} is based on merging ideas that appeared in recent papers by Huang et al.~\cite{huang2020approximability} and Huang and Ward~\cite{huang2020fpt} (see \cref{sec:multiple_matroids} for details). 
However, to obtain the same guarantee for non-monotone functions requires an interesting novel guessing scheme.
Moreover, this theorem cannot be improved by much.
The exponential dependence on $k$ is necessary by \cref{thm:hardness_multiple_streaming}, and the approximation ratio cannot be improved, even for a cardinality constraint, without using a linear in $n$ memory because of the inapproximability result of~\cite{feldman_2020_one-way}. It is open (even for a single partition matroid constraint) whether the exponential dependence on $\rho$ in \cref{th:matroid_intersection_alg_short} is necessary.

Up to this point, our inapproximability results concentrated on the asymptotic approximation ratio obtainable as a function of $k$, which is more relevant for large values of $k$. Small values of $k$ have also been considered extensively in the literature. For example, as aforementioned, the approximation ratio that can be obtained by a data streaming algorithm for a cardinality constraint, which is a special case of {\SMkM} with $k = 1$, was the subject of a long line of research~\cite{alaluf2020optimal,badanidiyuru2014streaming,feldman_2020_one-way,huang2020approximability,conficml0001MZLK19,mcgregor2019better} and is now essentially settled. Maximizing a monotone submodular function subject to a bipartite matching constraint is another important special case of {\SMkM}, this time for $k = 2$.

Since $\rho$ is usually polynomially related to $n$ in bipartite matching constraints, the class of algorithms considered interesting for the last problem is more restricted than for general {\SMkM}. Specifically, people are interested in algorithms that use $\tilde{O}(\rho)$ memory. That is, the algorithm does not use more memory (up to logarithmic factors) than what is required to simply store a solution. Such algorithms are known as semi-streaming algorithms. Recently, Levin and Wajc~\cite{levin2021streaming} described a semi-streaming algorithm for maximizing a monotone submodular function subject to a bipartite matching constraint which improves over the state-of-the-art for general {\SMkM} with a monotone objective function. They also proved, conditioned on some complexity-theoretic assumption, a lower bound of $1.914$ on the approximation ratio that can be obtained by a semi-streaming algorithm for the problem. Our final result improves over this upper bound and is independent of any complexity-theoretic assumption, but does assume that the graph can contain parallel edges (which are distinct elements from the point of view of the submodular objective function); the hardness of~\cite{levin2021streaming} applies even when this is not the case.
\begin{restatable}{theorem}{thmBipartiteMatching} \label{thm:bipartite_matching}
\iftoggle{InIntroduction}{No semi-streaming algorithm can obtain, with probability at least $2/3$, an approximation ratio of $2.692$ for maximizing a non-negative monotone submodular function subject to a bipartite matching constraint.}{Any single-pass data stream algorithm for {\MSMBM}  must use memory of $\Omega(r \log^{\omega(1)} r)$ on $r$-vertex graphs if it  finds a $2.692$-approximate solution with probability at least $2/3$.}
\end{restatable}
The last result is obtained by combining known hardness results for semi-streaming algorithms for the maximum cardinality matching problem~\cite{DBLP:conf/soda/Kapralov21} and submodular function maximization subject to a cardinality constraint~\cite{feldman_2020_one-way} in a non-trivial way so that the obtained hardness is stronger than what is known for any one of the two problems individually. Moreover, the reduction is general, and another consequence of it is the following: any semi-streaming algorithm for maximizing a monotone submodular function subject to a bipartite matching constraint that has an approximation guarantee better than $3$ would yield an improved semi-streaming algorithm for the maximum cardinality bipartite matching problem, which is a longstanding notorious open problem. We refer to reader to \cref{sec:main-hardness} for further detail.

\togglefalse{InIntroduction}

\section{Preliminaries}

In this section we give some additional technical details that are necessary for proving the results stated in \cref{sec:introduction}. In \cref{ssc:oracles} we discuss in more detail the oracles used to access the objective function and constraint matroids; and in \cref{ssc:chain} we present a known hard problem from which we often reduce to prove our inapproximability results.

\subsection{More details about the access oracles} \label{ssc:oracles}

As mentioned above, the matroid constraints are accessed via either a common independence oracle or $k$ distinct independence oracles, one per matroid. We also need to specify the method used to access the objective function $f$. In the submodular optimization literature, submodular objective functions such as $f$ are usually accessed via a \emph{value oracle} that, given a set $S \subseteq \cN$, returns $f(S)$. In the context of data stream and online algorithms, it is important that the independence and value oracles do not leak information in a way that contradicts our expectations from such algorithms. Accordingly, our algorithms query the oracles only on sets of elements that are explicitly stored in their memory.

The above information leakage issue often makes proving inapproximability results more complicated because such results have to formalize in some way the types of queries that are allowed (i.e., queries that are not considered ``leaky''). All our inapproximability results apply to the model used by our algorithmic results; namely, when the algorithm is allowed to query the oracles only on sets of elements that are explicitly stored in its memory---see~\cite{haung2021improvedunpublishedversion} for a formal statement of this natural model. However, we strive to weaken this assumption, and prove most of our inapproximability results even for algorithms that enjoy a less limited access to the oracles. For example, the proof of \cref{thm:hardness_multiple_streaming} manages to avoid this issue completely using the following technique. The algorithm is given upfront a ``super ground set'' and fully known objective function and matroids over this super ground set. The real ground set is then chosen as some subset of this super ground set, and only the elements of this real ground set appear in the input stream of the algorithm. Since the challenge that the algorithm has to overcome in this case is to remember which elements belong to the real ground set, the oracles cannot leak important information to the algorithm, and therefore, we allow the algorithm unrestricted access to them.

The situation for \cref{thm:multiple_matroids_single_oracle_hardness} is a bit more involved because of the following observation. If the algorithm is allowed unrestricted access to the common independence oracle, then it can construct $k$ matroids $M_1, M_2, \dotsc, M_k$ that are consistent with this oracle, which makes the distinction between a single common independence oracle and $k$ independence oracles mute. Therefore, some restriction on the access to the common independence oracle must be used. The (arguably) simplest and most natural restriction of this kind is to allow the algorithm to query the common independence oracle only on subsets that do not include any elements that did not appear in the stream so far; and it turns out that this simple restriction suffices for the proof of \cref{thm:multiple_matroids_single_oracle_hardness} to go through.

It remains to consider our last inapproximability result, namely \cref{thm:bipartite_matching}. Here the elements of the ground set are edges, and the algorithm is given each edge in the form of its two endpoints. Therefore, there is no need for independence oracles. (Formally, this situation is equivalent to the case mentioned above in which there is a ``super ground set'' that is given upfront to the algorithm, and only part of this super ground set appears in the stream.) Unfortunately, preventing information leakage via the value oracle is more involved. For simplicity, the proof that we give in \cref{sec:main-hardness} assumes the same model that we use in our algorithmic results. However, our proof can be extended also to algorithms with a more powerful way to access the objective function $f$, such as the $p$-players model described in~\cite{feldman_2020_one-way}.

\subsection{The \texorpdfstring{{\chainPn}}{CHAINp(n)} problem} \label{ssc:chain}

Many of our inapproximability results use reductions to a (hard) problem named \chainPn, introduced by Cormode, Dark, and Konrad~\cite{cormode_2019_independent}, which is closely related to the Pointer Jumping problem (see~\cite{chakrabarti_2007_lower}). 
In this problem, there are $p$ players $P_1, \ldots, P_p$. For every $1 \leq i < p$, player $P_i$ is given a bit string $x^{i} \in \{0,1\}^n$ of length $n$, and, for every $2 \leq i \leq p$, player $P_i$ (also) has as input an index $t^i\in \{1, 2, \dotsc, n\}$. (Note that the convention in this terminology is that the superscript of a string/index indicates the player receiving it.) 
Furthermore, it is promised that either $x^{i}_{t^{i + 1}} = 0$ for all $1 \leq i < p$ or $x^{i}_{t^{i + 1}} = 1$ for all these $i$ values. 
We refer to these cases as the $0$-case and $1$-case, respectively. 
The objective of the players in \chainPn is to decide whether the input instance belongs to the $0$-case or the $1$-case. The first player, based on the input bit string $x^{1}$, sends a message $M^1$ to the second player. Any player $2 \leq i < p$ , based on the  message it receives from the previous player (i.e., $M^{i-1}$), the input bit string $x^{i}$ and index $t^i$, sends message $M^i$ to the next player. 
The last player, based on $M^{p-1}$ and $t^p$, decides if we are in the $0$-case or $1$-case. 
Each player has unbounded computational power and can use any (potentially randomized) algorithm. 
We refer to the collections of the algorithms used by all the players as a protocol. The success probability of a protocol is the probability that its decision is correct, and the communication complexity of a protocol is the size of the maximum message sent (i.e., maximum size of $M^1, \ldots, M^{p-1}$). In~\cite{feldman_2020_one-way}, the following lower bound was shown for the \chainPn problem, which is very similar to the lower bounds previously proved by~\cite{cormode_2019_independent}.

\begin{restatable}{theorem}{thmpindex}[Theorem 3.3 in~\cite{feldman_2020_one-way}] 
    For any positive integers $n$ and $p\geq 2$, any (potentially randomized) protocol for \chainPn with success probability of at least $\sfrac{2}{3}$ must have a communication complexity of at least $\sfrac{n}{36 p^2}$. Furthermore, this holds even when instances are drawn from a known distribution $D(p,n)$. 
    \label{thm:pindex_hardness}
\end{restatable}
The distribution $D(p,n)$ referred to by \cref{thm:pindex_hardness} is simply the uniform distribution over all $0$-case and $1$-case instances (see the definition of $D^p$ in Appendix~C of~\cite{feldman_2020_one-way}). 

\section{Inapproximability for Common Independence Oracle Only} \label{sec:inapproximability_oracle}

In this section, we prove \cref{thm:multiple_matroids_single_oracle_hardness}, which also serves as a warm-up for proof techniques to be discussed later.
\cref{thm:multiple_matroids_single_oracle_hardness} shows strong limits on what can be achieved if one only has access to a common independence oracle, i.e., given a set $S \subseteq \cN$, we can determine whether $S \in \bigcap_{i = 1}^k \cI_i$, but we do not have a way to determine whether $S \in \cI_i$ for a particular matroid $M_i=(\cN,\cI_i)$.
Even though, in later sections, we focus on having one oracle for each matroid, the common independence oracle model is sometimes very natural.
As aforementioned, one can still obtain an $O(k \log k)$-approximation for \SMkM using a semi-streaming algorithm even in this restricted model~\cite{haba2020streaming}. 
Our hardness result, \cref{thm:multiple_matroids_single_oracle_hardness}, implies that this approximation ratio for \SMkM is optimal up to a factor of $O(\log k)$.

We highlight that \cref{thm:multiple_matroids_single_oracle_hardness} assumes that the algorithm can query its oracle only about subsets of elements that it has already read from the input stream. This assumption can be somewhat relaxed, but it cannot be completely removed because an algorithm can construct matroids $M_1, M_2, \dotsc, M_k$ that are consistent with a given independence oracle for the intersection of these matroids if it is allowed to query this oracle on every subset of the ground set.

Throughout this section, the submodular function we aim at maximizing is the cardinality function, as stated in \cref{thm:multiple_matroids_single_oracle_hardness}.

\medskip

A key component in the proof of \cref{thm:multiple_matroids_single_oracle_hardness} is the construction of $k$ matroids leading to an \SMkM instance where many elements, when appearing in the stream, are indistinguishable from elements that appeared so far.
This allows us to ``hide'' a large cardinality common independent set consisting of elements that cannot be distinguished at the time they appear in the stream from other elements that are not part of a large common independent set.

A second key ingredient is to link this construction of bad \SMkM instances to the hard \chain{p}{n} problem, by showing how a general \chain{p}{n} problem can be rephrased as a hard \SMkM problem.
This second component is not needed if we only consider algorithms that store elements explicitly, and can only access elements (to be returned or used in a call to the common independence oracle or value oracle) that are stored explicitly in memory.
However, the connection to \chain{p}{n} allows for showing \cref{thm:multiple_matroids_single_oracle_hardness} in its stated form, which is more general and would also allow algorithms to use any way to encode the indices of elements seen so far.

To better highlight the different proof ideas, we first consider the case where the algorithm is required to explicitly store elements, which allows us to readily derive \cref{thm:multiple_matroids_single_oracle_hardness} (even with a slightly stronger memory lower bound) for this case.
This nicely highlights why our bad family of \SMkM instances allows for capturing very hard problems.
We then consider the more general setting assumed by \cref{thm:multiple_matroids_single_oracle_hardness}, where elements need not be stored explicitly in memory.
We start with the discussion of our bad \SMkM instances, which is common to both cases.

\subsection{Construction of \texorpdfstring{$k$}{k} matroids for hard \SMkM instances}\label{sec:badKMatroidsSingleOracle}

Our construction of $k$ matroids for hard {\SMkM} instances has one parameter $m\in \mathbb{Z}_{\geq 1}$ (apart from the number $k$ of matroids) and $k-1$ special elements $u_1,\ldots, u_{k-1}$ to be defined later.
The ground set $\cN$ of the matroids has size $n = m\cdot k$, and we partition $\cN$ into $k$ sets $\cN_1,\ldots, \cN_k$ of size $m$ each.
For simplicity of notation, we use, for $i\in [k]$, the shorthands $\cN_{\leq i} = \cup_{j=1}^i \cN_i$ and $\cN_{\geq i} = \cN \setminus \cN_{\leq i-1}$. 
For each $i\in [k-1]$, let $u_i\in \cN_i$ be an element in $\cN_i$ that we fix. (The elements $u_i$ will play the role of important hidden elements that will be hard to distinguish from other elements when they appear in the stream.)
The matroids $M_i = (\cN, \cI_i)$, for every $i\in [k]$, are partition matroids defined by
\begin{equation*}
\cI_i \coloneqq \mleft\{S\subseteq \cN \colon |S\cap (\cN_{\leq i}\setminus \{u_1,\ldots, u_{i-1}\})| \leq 1\mright\}\enspace.
\end{equation*}
When intersecting these $k$ matroids, the common independent sets 
\begin{equation*}
\cI \coloneqq \bigcap_{i=1}^k \cI_i
\end{equation*}
are all sets $S\subseteq\cN$ that can be described as follows.
There is an index $i\in [k]$ such that $S$ contains no elements of $\cN_{\geq i+1}$, one element of $\cN_i$, and an arbitrary subset of the elements $\{u_1,\ldots, u_{i-1}\}$.
In particular, maximum cardinality common independent sets have size $k$, and they are the sets consisting of $\{u_1,\ldots, u_{k-1}\}$ and one more (arbitrary) element of $\cN_k$.

In our hard {\SMkM} instances, the stream first reveals the elements of $\cN_1$ (in an order to be specified), then the elements of $\cN_2$, and so on.
As mentioned in \cref{thm:multiple_matroids_single_oracle_hardness}, the objective is to return a common independent set of maximum cardinality.
The hardness of these instances follows from the fact that many elements cannot be distinguished from each other when we see only a prefix of the stream.
Formally, elements cannot be distinguished in a set system if they are \emph{equivalent}, which is defined as follows.
\begin{definition}[Equivalent elements]
Let $E$ be a finite set and $\mathcal{F} \subseteq 2^E$. Two elements $e_1,e_2\in E$ are \emph{equivalent} in $\mathcal{F}$ if for every $U\subseteq E\setminus \{e_1,e_2\}$ we have $U\cup \{e_1\} \in \mathcal{F}$ if and only if $U\cup \{e_2\}\in \mathcal{F}$.
\end{definition}

Given this definition, we get that as long as only elements from $\cN_{\leq i}$ are revealed in the stream, one cannot distinguish between the elements in $\cN_i$.
\begin{observation}\label{obs:NiAreEquivalent}
For any $i\in [k]$, all elements of $\cN_i$ are equivalent in $\{S\in \cI\colon S \subseteq \cN_{\leq i} \}$.
\end{observation}

We now formalize why this construction of \SMkM instances leads to hard problems that need high communication complexity.

\subsection{Hardness under element storage model}\label{sec:singleOracleElemStorage}

As mentioned, we first consider the model where algorithms have to store elements explicitly, which we simply call the \emph{element storage model}.
We now show for the element storage model how \SMkM instances based on the construction of \cref{sec:badKMatroidsSingleOracle} allow for readily deriving a lower bound on the \emph{maximum storage size}, i.e., the maximum number of elements to be stored, of a $(k-\varepsilon)$-approximation for \SMkM that only has access to a common independence oracle.
The lower bound we obtain is slightly stronger (by a factor of $k^3$) than the memory lower bound in \cref{thm:multiple_matroids_single_oracle_hardness}, which holds for a more general model.
The following statement formalizes our result for the element storage model.
\begin{theorem} \label{thm:multiple_matroids_single_oracle_hardness_element_storage}
A data stream algorithm for {\SMkM} in the element storage model, whose only access to the matroids is via the common independence oracle, and with expected approximation ratio $k - \eps$ (for some $\eps > 0$), must admit a maximum storage size of $\Omega(\sfrac{\eps n}{k^3})$ many elements. This holds even when the task is to find a maximum size common independent set in $k$ partition matroids.
\end{theorem}
\begin{proof}
Consider a data stream algorithm for \SMkM in the element storage model that has access only to the common independence oracle and achieves an approximation ratio of $k-\varepsilon$ (for some $\varepsilon > 0$).
Let $\Gamma$ be the maximum storage size of the algorithm.

We consider $k$ matroids $M_1,\ldots, M_k$ as described in \cref{sec:badKMatroidsSingleOracle}, with $u_i$ being chosen to be an arbitrary element in $\cN_i$ for $i\in [k-1]$.
In our \SMkM instance, first the elements of $\cN_1$ reveal in uniformly random order, then the elements in $\cN_2$, and so on.
For $i\in [k-1]$, let $S_i \subseteq \cN_{\leq i}$ be the elements saved by the algorithm after the elements of $\cN_{\leq i}$ have been revealed in the stream.
By \cref{obs:NiAreEquivalent}, at this point of the procedure, the elements of $\cN_i$ cannot be distinguished by a common independence oracle.
Therefore, because the elements $\cN_i$ arrived in uniformly random order, we have
\begin{equation}\label{eq:boundUiInSi}
\Pr[u_i \in S_i] \leq \frac{\Gamma}{m}\quad \forall i \in [k-1]\enspace.
\end{equation}
Let $T\in \cI$ be the set returned by the algorithm.
Clearly, to have $u_i\in T$, the element $u_i$ must be in memory from the moment it appears in the stream until termination of the algorithm.
This is only possible if $u_i\in S_i$.
Thus,
\begin{equation*}
\frac{k}{k-\varepsilon} \leq \bE[|T|]
                        \leq 1 + \bE[|T\cap \{u_1,\ldots, u_{k-1}\}|]
                        \leq 1 + \sum_{i=1}^{k-1} \Pr[u_i \in S_i]
                        \leq 1 + k \cdot \frac{\Gamma}{m}\enspace,
\end{equation*}
where the first inequality follows from our assumption that the algorithm has an expected approximation factor of $k-\varepsilon$ and the observation that $k$ is the size of a maximum cardinality common independent set,
and the last inequality is due to~\eqref{eq:boundUiInSi}.
By reordering terms and using $n=m\cdot k$, we obtain the desired result
\begin{equation*}
\Gamma \geq \frac{m}{k}\cdot \frac{\varepsilon}{k-\varepsilon}
       \geq \frac{\varepsilon m}{k^2}
          = \frac{\varepsilon n}{k^3}\enspace.
\qedhere
\end{equation*}

\end{proof}

\subsection{Generalized hardness}

We now show our more general hardness result, \cref{thm:multiple_matroids_single_oracle_hardness}, which does not assume that elements need to be stored explicitly.
This hardness is based on combining the ideas highlighted in \cref{sec:singleOracleElemStorage} with a reduction from \chain{k}{m}.
(We write \chain{k}{m} instead of \chain{k}{n} because the parameter $m$ used in the chain problem will be different from the size $n$ of the ground set of the \SMkM problem we consider.)
Specifically, we assume the existence of an algorithm named $\ALG$ for {\SMkM} with an expected approximation guarantee of $k - \varepsilon$, and show that this leads to a protocol $\PRT$ for \chain{k}{m} whose communication complexity depends on the space complexity of $\ALG$. This allows us to translate the communication complexity lower bound that is known for protocols for \chain{k}{m} into a space complexity lower bound for $\ALG$.
To reduce from the chain problem, consider an instance of \chain{k}{m}.
For $i\in [m]$, let $x^i\in \{0,1\}^m$ be the bit string revealed to player $P_i$, and, for $i\in \{2,\ldots, k\}$, let $t^i\in [m]$ be the index revealed to player $P_i$.

Before describing $\PRT$ itself, let us first present a simpler protocol for $\chain{k}{m}$, which is given as \cref{alg:chainsubreduction} below. We show that \cref{alg:chainsubreduction} already allows for distinguishing the $0$-case from the $1$-case of $\chain{k}{m}$ with a probability that depends on $\varepsilon$; however this probability is in general strictly below the required success probability of $\sfrac{2}{3}$.
The final protocol $\PRT$ is then obtained by boosting the success probability of \cref{alg:chainsubreduction}, which, as we discuss later, can simply be achieved by running multiple copies of \cref{alg:chainsubreduction} in parallel.

\cref{alg:chainsubreduction} creates, based on the information revealed from the \chain{k}{m} instance, an \SMkM instance on the fly on which \ALG is called.
This \SMkM instance is based on our construction of $k$ matroids described in \cref{sec:badKMatroidsSingleOracle}.
More precisely, we first define $k$ matroids $M_1,\ldots, M_k$ on a ground set $\cN=\cN_1 \cup \ldots \cup \cN_k$ of size $n=k\cdot m$ as explained in \cref{sec:badKMatroidsSingleOracle}.
For each $i\in [k]$, denote the elements of $\cN_i$ by $\cN_i = \{e^i_1,\ldots, e^i_m\}$, where we think of element $e^i_j$ as corresponding to the bit string entry $x^i_j$.
The choice of the special elements $u_1,\ldots, u_{k-1}$ will be set to $u_i= e^i_{t_{i+1}}$ for $i\in [k-1]$.
Moreover, \cref{alg:chainsubreduction} only forwards elements $e^i_j\in \cN$ to \ALG with indices $i\in [k]$ and $j\in [m]$ such that $x^i_j=1$.

\floatstyle{ruled}
\newfloat{protocol}{th}{lop}
\floatname{protocol}{Protocol}
\begin{protocol}
\caption{Reduction from \chain{k}{m} to {\SMkM}} \label{alg:chainsubreduction}
\textbf{Algorithm of Player $i$}
\begin{algorithmic}[1]
    \If{$i = 1$}
			\State Initialize a new instance of $\ALG$.
		\Else
			\State Create an instance of $\ALG$ that starts from the state forwarded by the previous player.
			\State Set $u_{i-1} = e^{i-1}_{t_{i}}$.
		\EndIf
		\If{$i < k$} 
			\For{$j = 1$ \textbf{to} $n$}
				\If{$b^i_j = 1$} Forward to $\ALG$ the element $e^i_j$. \EndIf
			\EndFor
			\State Forward to the next player the state of $\ALG$ and the values of $t_2, t_3, \dotsc, t_i$.\label{algline:forwardTIndices}
		\Else \ (i.e., if $i=k$) 
      \State Forward the elements in $\cN_k$ to \ALG in arbitrary order.
			\State Get the output set $T$ of \ALG. If $|T| \leq 1$ output ``$0$-case'', otherwise output ``$1$-case''.
		\EndIf
\end{algorithmic}
\end{protocol}

It is not difficult to verify that by the time the protocol arrives to player number $i$, this player has enough information to implement the common independence oracle of all the matroids restricted to the elements $\cN_{\leq i}$.
Indeed, this independence oracle only depends on the special elements $u_1,\ldots, u_{i-1}$ but not the special elements $u_i,\ldots, u_{k-1}$, which will be revealed (and set in \cref{alg:chainsubreduction}) at a later point in time.
This observation guarantees that algorithm $\ALG$ can be invoked in the way in which \cref{alg:chainsubreduction} uses it.

Note that the $k$-matroid intersection problem that \ALG receives consists of the $k$ matroids $M_1,\ldots, M_k$ restricted to the elements of
\begin{equation*}
\overline{\cN}\coloneqq \{e_j^i \colon i\in [k-1], j\in [m] \text{ with } x^i_j = 1\}\cup \cN_k \enspace.
\end{equation*}
Also note that the maximum cardinality common independent set within $\overline{\cN}$ is $1$ if we are in the $0$-case and $k$ in the $1$-case.

We begin the analysis of \cref{alg:chainsubreduction} by showing that it always succeed in the $0$-case.
\begin{lemma} \label{lem:chainreduction_zero_case}
In the $0$-case, \cref{alg:chainsubreduction} always produces the correct answer. 
\end{lemma}
\begin{proof}
As mentioned, in the $0$-case, any common independent set in $\overline{\cN}$ is of size at most $1$.
Because \ALG returns a common independent set $T\subseteq \overline{\cN}$, we have $|T|\leq 1$ and \cref{alg:chainsubreduction} identifies correctly that we are in the $0$-case.
\end{proof}

\begin{lemma} \label{lem:chainreduction_one_case}
In the $1$-case, \cref{alg:chainsubreduction} produces the correct answer with probability at least $\sfrac{\varepsilon}{k^2}$.
\end{lemma}
\begin{proof}
As discussed, the maximum cardinality common independent set in $\overline{\cN}$ in the $1$-case has size $k$. Because \ALG guarantees $(k-\varepsilon)$-approximation (in expectation), the expected size of its output $T$ fulfills
\begin{equation}\label{eq:lowBoundExpCardT}
\bE\left[|T|\right] \geq \frac{k}{k-\varepsilon}\enspace.
\end{equation}
Moreover, because $T$ has cardinality at most $k$, we obtain
\begin{align*}
\bE\left[|T|\right] \leq k \cdot \Pr\left[|T|>1\right] + \Pr\left[|T|\leq 1\right]
  = 1 + (k-1) \Pr\left[|T| > 1\right]\enspace.
\end{align*}
By combining the above relation with~\eqref{eq:lowBoundExpCardT}, we obtain, as desired,
\begin{equation*}
\Pr[|T| > 1] \geq \frac{\varepsilon}{k-\varepsilon} \cdot \frac{1}{k-1} \geq \frac{\varepsilon}{k^2}\enspace.
\qedhere
\end{equation*}
\end{proof}

We are now ready to present the promised protocol $\PRT$, which simply executes $\lceil \sfrac{2k^2}{\varepsilon} \rceil$ parallel copies of \cref{alg:chainsubreduction}, and then outputs that we are in the $1$-case if and only if at least one of the executions returned this answer.
We highlight that the parallel copies of \cref{alg:chainsubreduction} do not need to send independent copies of the values of $t_2,\ldots, t_i$ to the next player, as is done in Line~\ref{algline:forwardTIndices} of the protocol, because these indices are the same for each parallel run.
Hence, the protocol is run in parallel except for this one step.

\begin{corollary} \label{cor:chain_reduction}
In the $0$-case, $\PRT$ always answers correctly, and in the $1$-case, it still answers correctly with probability at least $\sfrac{2}{3}$.
\end{corollary}
\begin{proof}
The first part of the corollary is a direct consequence of Lemma~\ref{lem:chainreduction_zero_case}. Additionally, by Lemma~\ref{lem:chainreduction_one_case}, the probability that $\PRT$ answers that we are in the $0$-case when we are, in fact, in the $1$-case is at most
\begin{equation*}
    \left(1 - \frac{\varepsilon}{k^2}\right)^{\lceil \sfrac{2 k^2}{\varepsilon} \rceil}
    \leq
    \left(1 - \frac{\varepsilon}{k^2}\right)^{\sfrac{2 k^2}{\varepsilon}}
    \leq
    \left(e^{-\sfrac{\varepsilon}{k^2}}\right)^{\sfrac{2k^2}{\varepsilon}}
    =
    e^{-2}
    <
    \frac{1}{3}
    \enspace.
    \qedhere
\end{equation*}

\end{proof}

Combining \cref{cor:chain_reduction} with a communication complexity lower bound for the chain problem completes the proof of \cref{thm:multiple_matroids_single_oracle_hardness}.
\begin{proof}[Proof of \cref{thm:multiple_matroids_single_oracle_hardness}]
Because \cref{cor:chain_reduction} shows that $\PRT$ is an algorithm for \chain{k}{m} that succeeds with probability at least $\sfrac{2}{3}$, \cref{thm:pindex_hardness} guarantees that its communication complexity (i.e., the maximum size of any message in it) is at least $\Omega(\sfrac{m}{k^2})$.

Let us now bound the total communication complexity of $\PRT$ from above. We note that each message of $\PRT$ consists of $\lceil \sfrac{2k^2}{\varepsilon} \rceil$ states of $\ALG$ and up to $k - 1$ indices. (These are the indices $t_2,\ldots, t_k$ and we recall that only one copy of them is sent to the next player instead of sending one for each parallel execution of \cref{alg:chainsubreduction}.) Therefore, the total communication complexity of $\ALG$ is $O((\sfrac{k^2}{\eps}) \cdot S + k \cdot \log m)$, where $S$ is the space complexity of $\ALG$.

Combining the upper and lower bounds we have found for the communication complexity of $\PRT$, we get
\begin{equation}\label{eq:combProtocolBounds}
	O\left(\frac{k^2}{\eps} \cdot S + k \cdot \log m\right) = \Omega\left(\frac{m}{k^2}\right)
	\enspace.
\end{equation}
We now need to consider two cases. If $S = \Omega(\sfrac{\eps \log m}{k})$, then the $\tfrac{k^2}{\eps} \cdot S$ term on the left hand side of the last equality dominates, and thus, $S = \Omega(\sfrac{\eps m}{k^4}) = \Omega(\sfrac{\eps n}{k^5}) = \Omega(\sfrac{\eps n}{k^5 \log k})$, as desired.

It remains to consider the case of $S = O(\sfrac{\eps \log m}{k})$. In this case the second term on the left hand side of Equation~\eqref{eq:combProtocolBounds} dominates, which implies $k = \Omega(\sqrt[3]{\sfrac{m}{\log m}})$, and therefore,
\[
	\frac{\eps n}{k^5 \log k}
	=
	\frac{\eps m}{k^4 \log k}
	=
	\frac{\eps m}{\Omega(\frac{m}{\log m}) \cdot k \log k}
	=
	O(1)
	=
	O(S)
	\enspace,
\]
where the penultimate equality holds since $\eps \leq k - 1 < k$ and $\log k = \Omega(\log m - \log \log m) = \Omega(\log m)$.
\end{proof}

\section{Inapproximability for Multiple Independence Oracles} \label{sec:hardness_intersection}
In this section, we prove \cref{thm:hardness_multiple_streaming}, which gives a strong inapproximability result for data stream algorithms as a function of the number $k$ of matroids, even in the case when the objective function $f$ is a linear function (unlike in the previous section, here we allow access to the independence oracles of the individual matroids).
\thmHardnessMultipleStreaming*
Note that the above result implies that
\begin{enumerate*}[label=(\roman*)]
\item any data stream algorithm with an approximation guarantee  $o(k/\log k)$ requires super polynomial memory in $k$, and
\item any data stream algorithm with constant approximation guarantee requires exponential memory in $k$.
\end{enumerate*}

The techniques we use also readily imply hardness for the (preemptive) online version of this problem.
In this version, the elements of the ground set $\cN$ arrive online, and upon receiving each element the algorithm has to decide either to add this element to the solution it maintains, or to reject the element. 
If the algorithm accepts an element to its solution, it may remove this element from the solution at a later point; however, a decision to reject an element (or remove it from the solution at a later time) is irrevocable. 
The algorithm is also required to keep its solution feasible at all times. We have the following hardness in this model.
\begin{theorem} \label{thm:inapproximability}
  For \(k\geq 2\), the competitive ratio of any online algorithm for {\SMkM} against an oblivious adversary is at least $\frac{k}{81 \ln k}$. This holds even when the task is to find a maximum size common independent set in $k$ partition matroids, and the common rank of these matroids is $O(k / \log k)$.
\label{thm:hardness_multiple_online}
\end{theorem}

The key building block for both (streaming and online) hardness results is  a  ``hard'' distribution of instances described in \cref{ssc:distribution}. This distribution  is then used to  prove \cref{thm:hardness_multiple_online,thm:hardness_multiple_streaming} in \cref{sec:multiple_hardness_online,sec:multiple_hardness_streaming}, respectively, in a similar way to the proof of \cref{thm:multiple_matroids_single_oracle_hardness}. 

\subsection{Description of hard distribution} \label{ssc:distribution}

Let $p$ be a non-negative integer parameter of the construction. Our instances are subsets of the set $\cN = [p]^k$, where we allow multiple elements with the same coordinates (i.e., ``multi-subsets''). A set $S \subseteq \cN$ is independent in the matroid $M_i$ (for any integer $i \in [k]$) if and only if no two elements of $S$ share the same value in coordinate number $i$ (in other words, $u_i \neq v_i$ for every two distinct elements $u, v \in S$). 
One can observe that this definition makes $M_i$ a partition matroid. We recall that $S\subseteq \cN$ is a \emph{common independent set} if it is independent in all matroids; otherwise, we will refer to it as \emph{dependent}. 
Note that the common rank of the $k$ matroids $M_1,\ldots, M_k$ is $p$.

In \cref{alg:streaming_adversary} we describe a procedure for sampling $p-1$ many subsets $S_1, S_2, \ldots, S_{p-1} \subseteq [p]^k$ and $p-1$ ``hidden'' optimal elements $o^{(1)} \in S_1, o^{(2)}\in S_2, \ldots, o^{({p-1})}\in S_{p-1}$. 
 In every iteration $r=1,\ldots, p-1$, the algorithm first forms  $S_r$ by sampling $m$ elements independently and uniformly from those elements that form a common independent set with $\{o_1, \ldots, o_{r-1}\}$. That is, $S_r$ contains $m$ uniformly random  samples with replacements from  
\[
	\{u \in \cN \mid  \forall_{1 \leq i < r, 1 \leq j \leq k}\; o^{(i)}_{j} \neq u_{j}\}\enspace.
\]
 Then, 
\emph{after} the selection of $S_r$, the algorithm samples $o^{(r)}$ uniformly at random among the $m$ elements in $S_r$. 

 \begin{algorithm}
 \caption{\textsc{Hard Instance Generation }$(p, m)$} \label{alg:streaming_adversary}
 \begin{algorithmic}[1]
 \For{$r = 1$ \textbf{to} $p-1$}
 	\State Obtain $S_r$ by sampling $m$ elements uniformly and with replacement from \[\{u \in \cN \mid  \forall_{1 \leq i < r, 1 \leq j \leq k}\; o^{(i)}_{j} \neq u_{j}\} \enspace.\] 
 	\State Select $o^{(r)}$ from $S_r$ uniformly at random.
 \EndFor
 \State Output $S_1, \ldots, S_{p-1}$ and $o^{(1)}, \ldots, o^{({p-1})}$.
 \end{algorithmic}
 \end{algorithm}
 We remark that the algorithm  with small probability  may sample the same element more than once when forming the set $S_r$. When this happens, we consider these samples to be unique elements on the stream (that are dependent). This allows us to simplify the notation in the following as each set $S_r$ is now guaranteed to contain exactly $m$ elements. Formally, this corresponds to extending the ground set $\cN$ by making $m$ copies $u^1, \ldots, u^m$ of each element $u\in \cN$, and whenever an element $u$ is sampled $i$ times, we include the copies $u^1, \ldots, u^i$.  

We refer to  \cref{alg:streaming_adversary} as ``\textsc{Hard Instance Generation}'' as it is the basic building block of our hardness results in both the online and streaming models. 
More specifically,  our hardness result for the  online model  is based on the (random) stream obtained  by first feeding the elements in $S_1$ (in any order), then $S_2$ (in any order), and so on until $S_{p-1}$ is fed. The intuition is that when the algorithm has only seen  the elements in $S_1, \ldots, S_i$, it has no information about the selection of $o^{(i)}$ and so any online algorithm is unlikely to have saved the element $o^{(i)}$. In addition, while $\{o^{(1)}, \ldots, o^{(p-1)}\}$ is a common independent set by construction, we prove  (see \cref{lemma:dist_success} below) that any other two elements are likely to be dependent. This creates the ``gap'' between the values of the solution $\{o^{(1)}, o^{(2)}, \ldots, o^{(p-1)}\}$ and a solution with any other elements, which in turn  yields the desired hardness result. For our hardness in the data stream model, we  forward a subset of the above-mentioned stream, and the difficulty for a low-space streaming algorithm  is to ``remember'' whether the special elements $o^{(1)}, o^{(2)}, \ldots, o^{(p-1)}$ appeared in the stream. This is formalized in the next sections.

We complete this section by proving that, with good probability, any large solution must contain the hidden elements $o^{(1)}, \ldots, o^{(p-1)}$.
\begin{definition}
  We say that the output of \cref{alg:streaming_adversary} is \emph{successful} if any two elements $e, f \in S_1 \cup S_2  \cup \cdots \cup S_{p-1} \setminus \{o_1, \ldots, o_{p-1}\}$ are dependent, i.e., there is a coordinate $i\in [k]$ such that $e_i = f_i$.
\end{definition}

\begin{lemma}
    The output of \cref{alg:streaming_adversary} is successful with  probability at least $1- \binom{p m}{2} e^{-k/p}$.
    \label{lemma:dist_success}
\end{lemma}
\begin{proof}
    Consider two elements $u \in S_{r_1} \setminus \{o_{r_1}\}$ and $v \in S_{r_2}$ with $r_1\leq r_2$. As $u \neq o_{r_1}$, each coordinate of $u$ equals that of $v$ with probability at least $1/(p-r_1+1) \geq 1/p$. Now, as there are $k$ coordinates, and each coordinate of $u$ is sampled independently at random, 
    \[
        \Pr[\mbox{$\{u,v\}$ is a common independent set}] \leq (1-1/p)^k \leq e^{-k/p}\,.
    \]
    The lemma now follows  by taking the union bound over all  possible pairs $u,v$; the number of such pairs is upper bounded by $\binom{(p-1) m}{2}$. 
\end{proof}


\newcommand{\sALG}{ALG\xspace}


\subsection{Hardness for online algorithms} 
\label{sec:multiple_hardness_online}

Let $p = \lceil k/(27 \ln k)\rceil+1$ and $m = k^3$. We will prove that the competitive ratio of any online algorithm is at least $k/(81 \ln k)$. We assume throughout that $k$ is such that  $k/(81 \ln k) > 1$. This is without loss of generality since the statement is trivial if $k/(81 \ln k) \leq 1$.   We shall consider the following distribution of instances. Run \cref{alg:streaming_adversary} with the parameters $p$ and $m$ to obtain sets $S_1, \ldots, S_{p-1}$ (and hidden elements $o^{(1)}, \ldots, o^{(p-1)}$), and construct an input stream in which the elements of $S_1$ appear first (in any order) followed by those in $S_2$ and so on until the elements in $S_{p-1}$ appear.  We will show that any deterministic online algorithm \sALG cannot be $((p-1)/3)$-competitive on this distribution of instances. \cref{thm:hardness_multiple_online} then follows via Yao's principle.

To analyze the  competitive ratio of \sALG, 
let $O$ be the event that the output of \sALG contains some element of $\{o^{(1)}, o^{(2)}, \ldots, o^{(p-1)}\}$, and let $S$ be the event that the instance is successful. Then, if we use $|\sALG|$ to denote the size of the independent set outputted by \sALG,
\begin{align*}
    \bE[|\sALG|] 
    & \leq \left(1 - \Pr[\neg O,  S]\right) \cdot k + 1\enspace,
\end{align*}
where the inequality holds because any solution has size at most $k$, and if the algorithm fails to identify any element in $\{o^{(1)}, \ldots, o^{(p-1)}\}$, then it can produce solution of size at most $1$ for a successful instance.

Note that since $k/(81 \ln k) > 1$ we have that $p \leq 3 \cdot k/(27 \ln k) = k/(9 \ln k)$.  
Now, by \cref{lemma:dist_success} and the selection of $p$ and $m$, we have 
 \[   \Pr[ S]  \geq 1- \binom{p m}{2} e^{-k/p} 
     \geq 1- k^8  \cdot e^{-9 \ln k}
    = 1- 1/k\enspace .
\]

To bound $\Pr[\neg O]$, note that the set $S_r$ contains no information about $o^{(r)}$ because $o^{(r)}$ is selected uniformly at random from $S_r$. 
Moreover, when all $m$ elements of $S_r$ have been inspected in the stream, \sALG keeps at most $p$ independent elements.  
Any one of these elements is $o^{(r)}$ with probability at most $p/m$.  In other words, the algorithm selects $o^{(r)}$ with probability at most $p/m$.   Hence, by the union bound,  the algorithm has selected any of  the $p-1$ elements $\{o^{(1)}, \ldots, o^{(p-1)}\}$ with probability at most $p/m \cdot (p-1) \leq 1/k$.   We thus have $\Pr[\neg O] \geq 1-1/k$, and together with the above proved inequality $\Pr[S] \geq 1-1/k$, we get via the union bound
\begin{align*}
    \bE[|\sALG|] 
     \leq \left(1 - \Pr[\neg O,  S]\right) \cdot k + 1 
     \leq (2/k) \cdot k + 1 = 3\enspace. 
\end{align*}
Since the stream contains a solution $\{o^{(1)}, \ldots, o^{(p-1)}\}$ of size $p-1$, this implies that \sALG is not better than $((p-1)/3)$-competitive, which in turn implies \cref{thm:hardness_multiple_online}.

\subsection{Hardness for streaming algorithms}
\label{sec:multiple_hardness_streaming}

Let \sALG be a data stream algorithm for finding a set of maximum cardinality subject to $k$ partition matroid constraints.  Further suppose that \sALG has the following properties:

\begin{itemize}
    \item \sALG uses memory $M$;
    \item \sALG outputs an $\alpha$-approximate solution with probability at least $2/3$, where $\alpha \leq \frac{k}{32 \ln k}$ (note that $\alpha$ is also lower bounded by $1$ since it is an approximation ratio).
\end{itemize}

Select $p= \lceil 3\cdot \alpha\rceil $, and let $m$  be the smallest power of two such that $m\geq e^{k/(8\alpha)}$. Note that this selection satisfies
\[
 \mbox{$p\in [3\cdot \alpha, 4\cdot \alpha]$} \enspace, \quad \mbox{$m \in [e^{k/(8\alpha)}, 2 e^{k/(8\alpha)}]$}\enspace, \qquad \mbox{and} \qquad   8 \cdot p\leq k \leq m \enspace.
\]
We will  use \sALG to devise a protocol for the \chainPm problem that succeeds with probability at least $2/3$ and has communication complexity at most $M + p \log_2 m$. Combining this reduction with \cref{thm:pindex_hardness} then yields \cref{thm:hardness_multiple_streaming}, i.e., that any such algorithm $\sALG$ must have a memory footprint $M$ that is at least $\Omega\left(e^{k/(8\alpha)}/{k^2}\right)$.


\subsubsection{Description of protocol}
We use \sALG to obtain \cref{alg:chainreductionmultiple} for \chainPm. The protocol consists of two phases: a precomputation phase that is independent of the \chainPm instance, followed by a description of the messages of the players.

\paragraph{Precomputation phase.}  In the {precomputation phase}, the players use shared random coins\footnote{We note that the hardness result  of \chainPm (\cref{thm:pindex_hardness}) holds when the players have access to public coins, i.e., shared randomness. This is proved, e.g., in Theorem 3.3 of~\cite{feldman_2020_one-way}. In general, Newman's theorem~\cite{DBLP:journals/ipl/Newman91} says that we can turn any public coin protocol into a private coin protocol with little (logarithmic) increase in communication.} to generate instances from the same distribution produced by \cref{alg:streaming_adversary} for all possible values of $t^2, \ldots, t^p\in [m]$ in an instance of \chainPm.
Specifically,  first a set $T$ is sampled from the same distribution as the set $S_1$ produced by \cref{alg:streaming_adversary}. 
The elements of $T$ are then randomly permuted. For $j\in[m]$, we let $T(j)$ denote the $j$-th element in the the obtained ordered set. The reason that the elements in $T$ are randomly permuted is to make sure that for any fixed $t^2\in [m]$, the element $T(t^2)$ is uniformly random, and thus, has the same distribution as $o^{(1)}$ in \cref{alg:streaming_adversary}.  
Following the choice of $S_1 = T$ and $o^{(1)}\in S_1$, \cref{alg:streaming_adversary} proceeds to sample $S_2$ to be $m$ random elements that are independent with respect to $o^{(1)}$. 
In the precomputation phase we do so for each possible element in $T$, i.e., we sample sets $T_{1}, T_2, \ldots, T_m$, one for each  possible choice of  $t^2\in [m]$.  
Then, for each such $T_{t^2}$ we sample $m$ sets $T_{t^2,1}, T_{t^2, 2}, \ldots, T_{t^2, m}$ for all possible choices of $t^3 \in [m]$ and so on. 
The sets constructed in the precomputation phase can thus naturally be represented by a tree, where each path from the root $T$ to a leaf corresponds to a particular choice of $t^2, t^3, \ldots, t^p \in [m]$.  
For $m=3$ and $p=3$, this tree is depicted in \cref{fig:precomputation}. The thick path corresponds to the case of $t_2 = 2$ and $t_3 = 3$. 

As described above, we randomly permute the sets so as to make sure that, for fixed $t^2, t^3, \ldots t^r \in [m]$, the distribution of $o^{(1)}$ is the same as that of $T(t^2)$ and, in general, the distribution of $o^{(i)}$ is the same as that of $T_{t^2, \ldots, t^i}(t^{i+1})$.  This gives us the following observation.
\begin{observation}
    Fix $t^2, t^3, \ldots, t^r\in [m]$. Over the randomness of the precomputation phase, the elements $o^{(1)} = T(t^2), o^{(2)} = T_{t^2}(t^3), \ldots, o^{(r-1)} = T_{t^2, \ldots, t^{r-1}}(t^r)$ and the sets $S_1 = T, S_2 = T_{t^2}, \ldots, S_{r-1} = T_{t^2, \ldots, t^{r-1}}$ have the same distribution as  the output of \cref{alg:streaming_adversary}.
    \label{obs:same_dist}
\end{observation}
\begin{figure}[t]
\begin{center}
    \begin{tikzpicture}
        \node (s) at (0,0) {\small $T$}; 
        \node (s1) at (-4, -1.5) {\small $T_1$};
        \node (s2) at (0, -1.5) {\small $T_2$};
        \node (s3) at (4, -1.5) {\small $T_3$};
        
        \node (s11) at (-5, -3) {\small $T_{1,1}$};
        \node (s12) at (-4, -3) {\small $T_{1,2}$};
        \node (s13) at (-3, -3) {\small $T_{1,3}$};

        \node (s21) at (-1, -3) {\small $T_{2,1}$};
        \node (s22) at (-0, -3) {\small $T_{2,2}$};
        \node (s23) at (1, -3) {\small $T_{2,3}$};

        \node (s31) at (3, -3) {\small $T_{3,1}$};
        \node (s32) at (4, -3) {\small $T_{3,2}$};
        \node (s33) at (5, -3) {\small $T_{3,3}$};
        
        \draw (s) edge (s1) edge[ultra thick] (s2) edge (s3);
        \draw (s1) edge (s11) edge (s12) edge (s13);
        \draw (s2) edge (s21) edge (s22) edge[ultra thick] (s23);
        \draw (s3) edge (s31) edge (s32) edge (s33);
    \end{tikzpicture}
    \end{center}
    \caption{A tree representation of the sets computed during precomputation for $m = p =3$.}
    \label{fig:precomputation}
\end{figure}
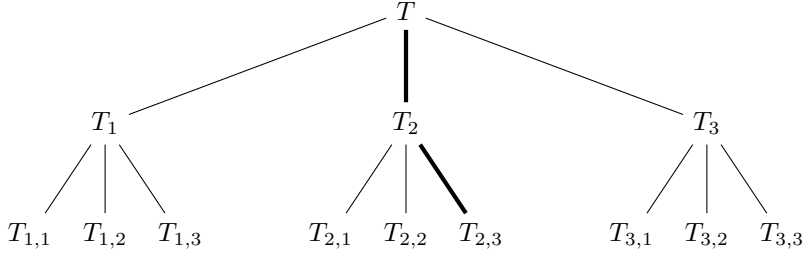

The reason the players do this precomputation  is that, after they have commonly agreed on the tree-structure of sets (which  can be generated using the public coins), it requires little communication to decide on a ``hard'' instance generated from the same distribution as \cref{alg:streaming_adversary}. Indeed, Player $r$ only needs to know $t^2, \ldots, t^{r}$ ($r\log_2 m$ bits of information) in-order to know the set $T_{t^2, \ldots, t^r}$. 

\paragraph{The messages of the players.}
After generating the (common) sets in the precomputation phase using the public coins, the players now proceed as follows.  The first player receives as input $x^1$ and simulates $ALG$ on the subset $\{T(j) \mid x^1_j = 1\}$ of $T$ corresponding to the $1$-bits. These elements are given to $ALG$ as a stream in any order. Player $1$ then sends to Player $2$ the message containing the state of $ALG$ after processing this stream of elements. 

The second player receives input $x^2, t^2$ and initializes \sALG with the state received from the first player. Then, the elements  $\{T_{t^2}(j) \mid x^2_j = 1 \}$  of $T_{t^2}$ that correspond to $1$-bits of $x^2$ are streamed to \sALG in any order.  Player $2$ sends to Player $3$ a message containing the state of $\sALG$  after processing these elements and the index $t^2$. 
Player $r,$ for $r = 3, \ldots, p-1$, proceeds similarly to Player $2$: given input $x^r, t^r$, \sALG is first initialized with the state received from the previous player, and then the elements $\{T_{t^2, \ldots, t^r}(j) \mid x^r_j = 1\}$ are streamed to \sALG in any order. Notice that Player $r$ knows $t^2, \ldots, t^{r-1}$ from the message of the previous player,  $t^r, x^r$ from the input, and $T_{t^2, \ldots, t^r}$ from the precompuation phase, and so the set $\{T_{t^2, \ldots, t^r}(j) \mid x^r_j = 1\}$ can be computed. Finally, Player $r$ sends to Player $r+1$ a message consisting of the indices $t^2, \ldots, t^r$ and the memory state of \sALG. 

The final player initializes \sALG with the received state and asks \sALG to return an independent set. If the independent set consists of at least two elements, Player $p$ outputs ``$1$-case'', and otherwise, the output is ``$0$-case''.
\begin{protocol}[t]
\caption{Reduction from \chainPm to {SMkM} } \label{alg:chainreductionmultiple}
\textbf{Precomputation}
\begin{algorithmic}[1]
    \State Let $T$ be a uniformly random subset of  $\cN = [p]^k$ of size $m$.
    \State Order the elements of $T$ randomly, and let $T(j)$ denote the $j$-th element.
    \For{$r = 2, \ldots,p-1$ and $t^2, \ldots, t^{r} \in [m]$}
    \State Identify $o^{(1)} = T(t^2)$,  $o^{(2)}=T_{t^2}(t^3)$, \ldots, $o^{(r-1)}=T_{t^2, \ldots, t^{r-1}}(t^r)$.
	\State Let $T_{t^2, \ldots, t^r}$ be a uniformly random subset of $\{u \in \cN \mid  \forall_{1 \leq i < r, 1 \leq j \leq k}\; o^{(i)}_{j} \neq u_{j}\}$ of size $m$.
    \State Order the elements of $T_{t^2, \ldots, t^r}$ randomly, and let $T_{t^2, \ldots, t^r}(j)$ denote the $j$-th element.
    \EndFor
\end{algorithmic}

\textbf{Player $P_r$'s Algorithm for $r=1, \ldots, p-1$}
\begin{algorithmic}[1]
    \State Initialize $ALG$ with the received memory state  (or initial state if first player).
    \State Simulate $ALG$ on the elements $T_r = \{T_{t^2, \ldots, t^{r}}(j) \mid j\in [n] \text{ with }x^{t}_j = 1\}$ given in any order.
    \State Send to $P_{r+1}$ the values $t^2, t^3, \ldots, t^r$ and the memory state of ALG. 
\end{algorithmic}
\textbf{Player $P_p$'s Algorithm}
\begin{algorithmic}[1]
    \State If $ALG$ with the received memory state returns an independent set of size at least $2$, output ``$1$-case''; otherwise, output ``$0$-case''.
\end{algorithmic}
\end{protocol}

\subsubsection{Analysis}
The messages sent by the players contain the memory state of \sALG and at most $p-2$ indices $t^2, \ldots, t^{p-1}\in [m]$. The memory state of \sALG is at most $M$ bits by assumption, and each index requires $\log_2 m$ bits. The communication complexity of the protocol is, therefore, upper bounded by $M +  p \log_2 m$. 

To analyze the success probability of the protocol, we have the following lemma.
\begin{lemma}
    The instance that the players stream to $\sALG$ satisfies the following:
    \begin{itemize}
        \item In the $1$-case, the stream contains $p-1$ independent elements.
        \item In the $0$-case, with probability at least $2/3$, any two elements in the stream are dependent.
    \end{itemize}
    \label{lemma:prot_correctness}
\end{lemma}
\begin{proof}
    In the $1$-case, we have $x^1_{t^2} = x^2_{t^3} = \ldots = x^{p-1}_{t^p} =  1$ and so the elements $T(t^2)$, $T_{t^2}(t^3)$, \ldots, $T_{t^2, \ldots, t^{p-1}}(t^p)$ are part of the stream. By definition, they form an independent set consisting of $p-1$ elements.
    
    In the $0$-case, we have that $x^1_{t^2} = x^2_{t^3} = \ldots = x^{p-1}_{t^p} =  0$ and so any two elements $e, f$ in the stream belong to the set $S_1 \cup S_2 \ldots, S_{p-1} \setminus \{o^{(1)}, o^{(2)},  \ldots, o^{(p-1)}$\}, where $S_1  = T, o^{(1)} = T(t^2)$ and $S_j = T_{t^2, \ldots, t^j}, o^{(j)} = T_{t^2, \ldots, t^j}(t^{j+1})$ for $j = 2, \ldots, p-1$. By \cref{obs:same_dist}, we can apply \cref{lemma:dist_success} to obtain that, with probability at least $1-\binom{p m}{2} e^{-k/p}$, any two elements in the stream are dependent. The statement now follows  since the selection of our parameters $p, m$ and $k$ implies
    \[
        \binom{p m}{2} e^{-k/p}    \leq m^4 e^{-k/p}  \leq 16e^{k/(8\alpha)} \cdot e^{-k/(4\alpha)} = 16 e^{-k/(8\alpha)}\leq 1/3\enspace, 
    \]
    where for the first inequality we used $p\leq m$, and for the second inequality we used that $p \leq  4\cdot \alpha$ and $m \leq 2e^{k/(8\alpha)}$. The last inequality holds for $k \geq 3$ (the case of $k = 2$ can be ignored because it implies $k / (32 \ln k) < 1$, which makes \cref{thm:hardness_multiple_streaming} trivial).
\end{proof}

We now argue how the above lemma implies that \cref{alg:chainreductionmultiple} has a success probability of $2/3$ in both the  $0$-case and $1$-case. For $0$-case instances, we have with probability $2/3$ that any two elements are dependent. Hence, with that probability,  there is no way for \sALG to return an independent set with more than one element. Thus, the output of Player $p$ is correct with probability at least $2/3$ in the $0$-case. In the $1$-case, the stream always contains a solution of value $p-1$. By the assumption that $\sALG$ returns an $\alpha$-approximate solution with probability at least $2/3$, $\sALG$ returns an independent set of size at least $(p-1)/ \alpha$ with probability at least $2/3$. This implies that Player $p$ is correct in this case with probability $2/3$ since
\[
    (p-1) /\alpha    \geq (3\cdot\alpha-1) /\alpha  \geq  2\enspace. 
\]
Using \sALG we have, thus, devised a protocol for \chainPm that is correct with probability $2/3$ and has a communication complexity that is upper bounded by $M + p \log_2 m$. By \cref{thm:pindex_hardness}, we thus must have $(M + p \log_2 m) \geq m/(36p^2)$. Now using that $p \leq k/8$ and    $e^{k /(8\alpha)}  \leq m$, we get 
\[
    M \geq \frac{m}{36 p^2} - p \log_2 n 
    \geq \frac{64}{36} \frac{m}{ k^2} - k^2
    \geq \frac{64}{36} \frac{e^{k/(8\alpha)}}{k^2} - k^2\enspace,  
    \]
which is $\Omega\left(e^{k/(8\alpha)}/{k^2}\right)$ since by assumption on $\alpha$ we have $e^{k/(8\alpha)} \geq k^4$. We have thus proved that the memory usage  $M$ of \sALG must be at least $\Omega\left(e^{k/(8\alpha)}/{k^2}\right)$ as required by \cref{thm:hardness_multiple_streaming}.

\section{\texorpdfstring{$2$}{2}-Approximation Data Stream Algorithm} \label{sec:multiple_matroids}

We begin this section with a more detailed statement of \cref{th:matroid_intersection_alg_short}. Throughout this section, the $\tilde{O}$ notation hides a $\log n$ factor. If one makes (as is done in \cref{sec:introduction}) the simplifying assumption that an element of $\cN$ can be stored in $O(1)$ space, then $\tilde{O}$ can be replaced with $O$ throughout the section.

\thMatroidIntersectionAlgShort*
\noindent \textbf{Remark:} If the rank $\rho_I$ of the \emph{matroid intersection} (i.e., the maximum size of a common independent set) is known, then one can truncate all the matroid constraints to this rank before executing the algorithm whose existence is guaranteed by \cref{th:matroid_intersection_alg_short}. This does not affect the approximation ratio guaranteed by the algorithm, but makes the space complexity depend on $\rho_I$ instead of $\rho$.

To make our main new ideas easier to understand, we prove in this section a simplified version of \cref{th:matroid_intersection_alg_short}, which appears below as \cref{th:matroid_intersection_alg_simplified}. \Cref{th:matroid_intersection_alg_simplified} assumes access to two values $\tau$ and $|OPT|$ that are not usually available, where the set $OPT$ represents an arbitrary (but fixed) optimal solution for the problem. A proof of the original \cref{th:matroid_intersection_alg_short} can be found in \cref{apx:matroid_intersection_alg}. In this proof, we compensate for the lack of knowledge of $\tau$ and $|OPT|$ using a standard technique originally due to~\cite{badanidiyuru2014streaming} (our implementation of this technique adopts some details from~\cite{haba2020streaming}).

\begin{theorem} \label{th:matroid_intersection_alg_simplified}
For every $\eps \in (0, 1/7)$, there exists a $(2 + O(\eps))$-approximation data stream algorithm for {\SMkM} that has a space complexity of $\tilde{O}((\frac{k^{2\rho k}(\log \rho + \log \eps^{-1})}{\eps^3})^{\rho + 1}\rho)$ and assumes access to $|OPT|$ and a value $\tau \in [f(OPT), 3f(OPT)]$. If the objective function of {\SMkM} is guaranteed to be monotone, the space complexity of the algorithm can be improved to $\tilde{O}((\frac{k^{\rho k}(\log \rho + \log \eps^{-1})}{\eps})^{\rho + 1}\rho)$.
\end{theorem}

The proof of \cref{th:matroid_intersection_alg_simplified} begins by defining an auxiliary problem. We say that a data stream algorithm marks online a set $U$ if, upon reading an element $u$ from the stream, the algorithm decides immediately (and without observing future elements of the stream) whether $u$ belongs to $U$. Given an instance of {\SMkM}, a data stream algorithm for the {\FESMkM} problem is required to mark online a set $U$ of elements of the stream that has the following property. For \emph{every} common independent set $O \subseteq \cN$, if $o_1$ is the first element of $O$ to arrive in the stream, then the set $U$ must include some element $u \in U$ such that
\begin{enumerate}[label=({A}\arabic*),leftmargin=2.5em]
	\item $u$ is either equal to $o_1$ or arrived before $o_1$ in the stream, and \label[property]{condition:u_weakly_before}
	\item $(O \setminus \{o_1\}) \cup \{u\}$ is a common independent set.
\end{enumerate}
We stress that the set $U$ should have the above properties with respect to all common independent sets $O \subseteq \cN$ at the same time. Additionally, to simplify statements like the one appearing in \cref{condition:u_weakly_before}, below we say that an element $u$ appears \emph{weakly before} (respectively, \emph{weakly after}) an element $v$ if $u$ is either equal to $v$ or appears before (respectively, after) $v$.

A trivial algorithm for {\FESMkM} is an algorithm that marks all elements of the stream as belonging to $U$. However, we are interested in algorithms for {\FESMkM} that mark only a (relatively) small number of elements. Specifically, in \cref{sec:first_element_SMkM} we use ideas from~\cite{huang2020fpt} to prove the following theorem.
\begin{restatable}{theorem}{thmFirstElement} \label{thm:first_element}
There exists an algorithm (named {\FEA}) for {\FESMkM} that marks $O(k^{\rho k - 1})$ elements as belonging to $U$, and has a space complexity of $\tilde{O}(k^{\rho k + 1}\rho)$.
\end{restatable}

We now get to the proof of \cref{th:matroid_intersection_alg_simplified} for the case of a monotone objective function.
The algorithm that we use for proving \cref{th:matroid_intersection_alg_simplified} in this case is given as \cref{alg:monotone_SMkM}. In this algorithm we denote by $o_1, o_2, \dotsc, o_{|OPT|}$ the elements of $OPT$ in the order of their arrival. In iteration number $j$ of the algorithm, the algorithm aims to add to its solution either the element $o_j$ itself, or an appropriate replacement, which the algorithm does with the aid of {\FEA}. The pseudocode of \cref{alg:monotone_SMkM} requires guessing values and pieces of information that the algorithm has no access to. In general, this guessing must be done by trying all the possible options in parallel, and we discuss this in more detail in the proof of \cref{obs:monotone_space_complexity}. 

\begin{algorithm}
\caption{Algorithm for {\SMkM} with a Monotone Objective Function} \label{alg:monotone_SMkM}
\begin{algorithmic}[1]
\State $S_0 \gets \varnothing$.
\For{$j = 1$ \textbf{to} $|OPT|$}
	\State \parbox[t]{150mm}{Guess an integer $0 \leq i(j) \leq \lceil \log_{1 + \eps} (j^2/\eps) \rceil$ such that $f(o_{j} \mid S_{j - 1}) \in (\tau / (1 + \eps)^{i(j) + 1}, \tau / (1 + \eps)^{i(j)}]$. If there is no such integer $i(j)$, guess $i(j) = \infty$.\strut} \label{line:guess_i}
	\If{$i(j) \neq \infty$}
		\State \parbox[t]{145mm}{Initialize an instance $ALG_j$ of {\FEA} that uses $g(T) = f(T \mid S_{j-1})$ as the objective function, and the matroids $M_1 / S_{j-1}, M_2 / S_{j-1}, \dotsc, M_k / S_{j-1}$ as the constraints.\footnotemark\strut}
		\For{every element $u$ that arrives in the input stream} \label{line:elements_reading}
			\If{$f(u \mid S_{j - 1}) \in (\tau / (1 + \eps)^{i(j) + 1}, \tau / (1 + \eps)^{i(j)}]$}
				\State Feed $u$ to $ALG_j$.
				\If{$ALG_j$ marks $u$}
					\State \parbox[t]{128mm}{Guess whether $u$ is the first element marked by $ALG_j$ such that $S_{j - 1} \cup \{o_{j + 1}, o_{j + 2}, \dotsc, o_{|OPT|}\} \cup \{u\}$ is a common independent set. If this is the case, denote $u$ by $u_j$ and exit the loop starting on \cref{line:elements_reading}.\strut} \label{line:u_j_choose}
				\EndIf
			\EndIf
		\EndFor
		\State Set $S_{j + 1} \gets S_{j} \cup \{u_{j}\}$.
	\Else
		\State Set $S_{j + 1} \gets S_{j}$.
	\EndIf
\EndFor
\State \Return{$S_{|OPT|}$}.
\end{algorithmic}
\end{algorithm}

\footnotetext{The operator `$/$' represents matroid contraction here. For example, $M_1 / S_{j - 1}$ is a matroid over the ground set $\cN \setminus S_{j - 1}$ in which a set $T \subseteq \cN \setminus S_{j - 1}$ is independent if and only if $T \cup S_{j - 1}$ is independent in $M_1$.}The heart of the analysis of \cref{alg:monotone_SMkM} is the following lemma, which shows that \cref{alg:monotone_SMkM} is able to select an element $u_j$ in every relevant iteration. This is important since the element $u_j$ is a good replacement for $o_j$ both in terms of the matroid constraints (due to \cref{line:u_j_choose} of the algorithm) and in terms of the objective function (since its marginal contribution is similar to the one of $o_j$).
\begin{lemma} \label{lem:u_k_selection}
For every integer $1 \leq j \leq |OPT|$, the set $S_j \cup \{o_{j + 1}, o_{j + 2}, \dotsc, o_{|OPT|}\}$ is a common independent set. Furthermore, if $i(j) \neq \infty$, then \cref{alg:monotone_SMkM} selects some element as $u_j$, and this element appears weakly before $o_j$ in the input stream.
\end{lemma}
\begin{proof}
We prove the lemma by induction on $j$, and the induction hypothesis that we use is that the lemma holds for all lower values of $j$. This hypothesis immediately implies that $S_{j-1} \cup \{o_{j}, o_{j + 1}, \dotsc, o_{|OPT|}\}$ is a common independent set.\footnote{For $j = 1$ this does not follow from the induction hypothesis, but it is still true since $S_{j-1} \cup \{o_{j}, o_{j + 1}, \dotsc, o_{|OPT|}\} = OPT$ in this case.} The lemma now follows in the case of $i(j) = \infty$ because in this case we have $S_j = S_{j - 1}$, and therefore, also
\[
	S_j \cup \{o_{j + 1}, o_{j + 2}, \dotsc, o_{|OPT|}\}
	\subseteq
	S_j \cup \{o_{j}, o_{j + 1}, \dotsc, o_{|OPT|}\}
	=
	S_{j-1} \cup \{o_{j}, o_{j + 1}, \dotsc, o_{|OPT|}\}
	\enspace,
\]
which guarantees that $S_j \cup \{o_{j + 1}, o_{j + 2}, \dotsc, o_{|OPT|}\}$ is a common independent set.

It remains to consider the case of $i(j) \neq \infty$. Let $u_{\text{first}}$ be the first element in the remaining data stream at the beginning of iteration number $j$ of \cref{alg:monotone_SMkM}. One can note that $u_{\text{first}}$ must be either the first element of the original input stream or the element following $u_{j'}$ for some integer $1 \leq j' < j$ for which $i(j') \neq \infty$. This observation implies that $o_j$ appears weakly after $u_{\text{first}}$ in the input stream because the induction hypothesis guarantees that $u_{j'}$ appears weakly before $o_{j'}$ in the stream. At this point we would like to argue that $ALG_j$ is guaranteed to pick an element $u'_j$ that appears weakly before $o_j$ and obeys the condition on \cref{line:u_j_choose} of \cref{alg:monotone_SMkM}, and, therefore, can be marked as $u_j$. Unfortunately, this does not follow immediately from \cref{thm:first_element} since the input stream fed to $ALG_j$ varies based on the output this algorithm.

To solve the above difficulty, we need to consider an input stream $\cI$ that is a sub-stream of the original input stream. Specifically, an element $u$ belongs to $\cI$ if (1) it appears weakly after $o_j$, or (2) it appears weakly after $u_{\text{first}}$ and also obeys $f(u \mid S_{j - 1}) \in (\tau / (1 + \eps)^{i(j) + 1}, \tau / (1 + \eps)^{i(j)}]$. By the induction hypothesis, $\{o_j, o_{j + 1}, \dotsc, o_{|OPT|}\}$ is a subset of the stream $\cI$ that is independent in all the matroids $M_1 / S_{j - 1}, M_2 / S_{j - 1}, \dotsc, M_k / S_{j - 1}$. Therefore, given the input stream $\cI$, $ALG_j$ is guaranteed to mark some element $u'_j$ that arrives weakly before $o_j$ and also obeys that
\[
	S_{j - 1} \cup \{o_j, o_{j + 1}, \dotsc, o_{|OPT|}\} - o_j + u'_j
	=
	S_{j - 1} \cup \{o_{j + 1}, \dotsc, o_{|OPT|}\} + u'_j
\]
is a common independent set. In reality, $ALG_j$ does not get the stream $\cI$. However, since $ALG_j$ makes its decisions online, it will mark $u'_j$ if given any prefix of $\cI$ that includes $u'_j$. This implies that $u'_j$ is marked by $ALG_j$ unless some previous element is chosen as $u_j$. Therefore, \cref{alg:monotone_SMkM} is guaranteed to select some element as $u_j$, and this element appears weakly before $u'_j$, and therefore, also weakly before $o_j$.

To complete the proof of the lemma, we note that, by the definition of $u_j$,
\[
	S_{j - 1} \cup \{o_{j + 1}, o_{j + 2}, \dotsc, o_{|OPT|}\} + u_j
	=
	S_j \cup \{o_{j + 1}, o_{j + 2}, \dotsc, o_{|OPT|}\}
\]
is a common independent set.
\end{proof}

Using the last lemma, we can now prove the approximation ratio of \cref{alg:monotone_SMkM}.
\begin{corollary} \label{cor:approximation_monotone}
\cref{alg:monotone_SMkM} is a $(2 + O(\eps))$-approximation algorithm that outputs a set $S_{|OPT|}$ that is a common independent set.
\end{corollary}
\begin{proof}
The second part of the corollary follows immediately from \cref{lem:u_k_selection} by plugging in $j = |OPT|$. Therefore, we concentrate on proving the first part of the corollary.

First, we note that for every integer $1 \leq j \leq |OPT|$ the submodularity of $f$ guarantee that $f(o_j \mid S_{j - 1}) \leq f(\{o_j\}) \leq f(OPT) \leq \tau$ (the second inequality follows from the optimality of $OPT$). Therefore, if $i(j) = \infty$, then we must have
\[
	f(o_j \mid S_{j - 1})
	\leq
	\frac{\tau}{(1 + \eps)^{\lceil \log_{1 + \eps}(j^2/\eps) + 1 \rceil}}
	\leq
	\frac{\eps \tau}{j^2}
	\enspace,
\]
which implies
\[
	f(S_j) - f(S_{j - 1})
	=
	0
	\geq
	\frac{f(o_j \mid S_{j - 1})}{1 + \eps} - \frac{\eps \tau}{j^2}
	\enspace.
\]
Additionally, if $i(j) \neq \infty$, then the fact that $u_j$ is one of the elements fed into $ALG_j$ guarantees
\[
	f(S_j) - f(S_{j - 1})
	=
	f(u_j \mid S_{j - 1})
	\geq
	\frac{\tau}{(1 + \eps)^{i(j) + 1}}
	\geq
	\frac{f(o_j \mid S_{j - 1})}{1 + \eps}
	\enspace.
\]

Combining the last two inequalities, we get that for every integer $1 \leq j \leq |OPT|$,
\begin{align*}
	f(S_j) - f(S_{j - 1})
	\geq{} &
	\frac{f(o_j \mid S_{j - 1})}{1 + \eps} - \frac{\eps \tau}{j^2}
	\geq
	\frac{f(o_j \mid S_{j - 1} \cup \{o_{j + 1}, o_{j + 2}, \dotsc, o_{|OPT|}\})}{1 + \eps} - \frac{\eps \tau}{j^2}\\
	={} &
	\frac{f(\{o_{j}, o_{j + 1}, \dotsc, o_{|OPT|}\} \cup S_{j - 1}) - f(\{o_{j + 1}, o_{j + 2}, \dotsc, o_{|OPT|}\} \cup S_{j - 1})}{1 + \eps} - \frac{\eps \tau}{j^2}\\
	\geq{} &
	\frac{f(\{o_{j}, o_{j + 1}, \dotsc, o_{|OPT|}\} \cup S_{j - 1}) - f(\{o_{j + 1}, o_{j + 2}, \dotsc, o_{|OPT|}\} \cup S_{j})}{1 + \eps} - \frac{\eps \tau}{j^2}
	\enspace,
\end{align*}
where the second inequality follows from the submodularity of $f$, and the last inequality follows from its monotonicity. Adding up the last inequality over all $1 \leq j \leq |OPT|$ now yields
\begin{align*}
	f(S_{|OPT|})
	\geq
	f(S_0) + \frac{f(OPT) - f(S_{|OPT|})}{1 + \eps} - \sum_{i = 1}^{|OPT|} \frac{\eps \tau}{j^2}
	\geq
	\frac{f(OPT) - f(S_{|OPT|})}{1 + \eps} - 6\eps \cdot f(OPT)
	\enspace,
\end{align*}
where the second inequality holds by the non-negativity of $f$. Rearranging this inequality now completes the proof of the lemma since $\eps \in (0, 1/7)$.
\end{proof}

To complete the proof of \cref{th:matroid_intersection_alg_simplified} for the case of a monotone objective function, it remains to analyze the space complexity of \cref{alg:monotone_SMkM}, which we do next.
\begin{observation} \label{obs:monotone_space_complexity}
\Cref{alg:monotone_SMkM} can be implemented so that its space complexity is
\[
	\tilde{O}(\eps^{-\rho}k^{\rho k(\rho + 1)}\rho(\log \rho + \log \eps^{-1})^{\rho + 1})
	=
	\tilde{O}\mleft(\mleft(\frac{k^{\rho k}(\log \rho + \log \eps^{-1})}{\eps}\mright)^{\rho + 1}\rho\mright)
	\enspace.
\]
\end{observation}
\begin{proof}
To implement the guesses in \cref{alg:monotone_SMkM}, one has to fork the execution at each guess. In other words, if the guess has $Y$ possible outcomes, then the ``thread'' of the algorithm performing the guess should split into $Y$ parallel threads at the point of guess, one executing with each one of the possible outcomes, which guarantees that at least one thread terminates after making all the correct guesses. Then, the algorithm should output the best among all the feasible solutions produced by all its threads, which guarantees that the output solution is at least as good as the solution produced by the above mentioned thread that makes all the right guesses.

Let us now bound the number of threads used in the above implementation. \cref{alg:monotone_SMkM} makes guesses in two lines only. In \cref{line:guess_i} the algorithm guesses a value for $i(j)$. The number of options for this value is
\[
	1 + \lceil \log_{1 + \eps} (j^2/\eps) \rceil
	\leq
	1 + \lceil \log_{1 + \eps} (\rho^2/\eps) \rceil
	=
	O(\eps^{-1}(\log \rho + \log \eps^{-1}))
	\enspace.
\]
Then, in \cref{line:u_j_choose} the algorithm guesses the first element in the output of $ALG_j$ that has a property which at least one element in this output is guaranteed to have by \cref{lem:u_k_selection}. Since the output of $ALG_j$ is of size $O(k^{\rho k - 1})$ by \cref{thm:first_element}, this is the number of possible guesses in this line. Since each iteration of \cref{alg:monotone_SMkM} makes a single guess of each of the above kinds, the number of threads used by the algorithm increases by a factor of $O(\eps^{-1}k^{\rho k - 1}(\log \rho + \log \eps^{-1}))$ following each such iteration.\footnote{In the case of \cref{line:u_j_choose} it might look like the algorithm makes multiple guesses in every iteration since this line is technically a part of the inner loop starting on \cref{line:elements_reading}; however, one can observe that it suffices to guess only a single number per iteration of the algorithm---namely, the index of the correct element within the output of $ALG_j$---in order to implement this line.}

Since \cref{alg:monotone_SMkM} makes $|OPT| \leq \rho$ iterations, it has only $O(\eps^{-\rho}k^{\rho(\rho k - 1)}(\log \rho + \log \eps^{-1})^\rho)$ threads upon termination (which is the moment in which the number of threads is the largest). Each one of these threads requires enough space to store the index $i(j)$, $O(\rho)$ elements and the state of $ALG_j$, which by \cref{thm:first_element} can be done with a space complexity of
\[
	O(\log \eps^{-1} + \log \log \rho) + \tilde{O}(\rho) + \tilde{O}(k^{k\rho + 1} \rho)
	=
	\tilde{O}(\log \eps^{-1} + k^{k\rho + 1} \rho)
	\enspace.
\]
Multiplying the last bound with the bound we have on the number of threads used by \cref{alg:monotone_SMkM} completes the proof of the observation.
\end{proof}

We now get to the proof of \cref{th:matroid_intersection_alg_simplified} in the general case in which the objective function is not guaranteed to be monotone. 
The algorithm that we use for proving \cref{th:matroid_intersection_alg_simplified} in this case is given as \cref{alg:non-monotone_SMkM}. The only difference between this algorithm and \cref{alg:monotone_SMkM} is that we now execute $ALG_j$ multiple times in each iteration, and pick as $u_j$ a single element out of the union of the outputs of these executions. This allows us to either pick $o_j$ itself or have multiple possible candidates as replacements for it. Out of these candidates, the algorithm picks a uniformly random one in \cref{line:random_decision}. Interestingly, as is explained in the proof of \cref{obs:non-monotone_space_complexity}, we end up implementing this random choice by trying all the possible options rather than by making a true random decision.

\begin{algorithm}[th]
\caption{Algorithm for {\SMkM} with a (not necessarily monotone) Objective Function} \label{alg:non-monotone_SMkM}
\begin{algorithmic}[1]
\State $S_0 \gets \varnothing$.
\For{$j = 1$ \textbf{to} $|OPT|$}
	\State \parbox[t]{150mm}{Guess an integer $0 \leq i(j) \leq \lceil \log_{1 + \eps} (j^2/\eps) \rceil$ such that $f(o_{j} \mid S_{j - 1}) \in (\tau / (1 + \eps)^{i(j) + 1}, \tau / (1 + \eps)^{i(j)}]$. If there is no such integer $i(j)$, guess $i(j) = \infty$.\strut} \label{line:guess_i_non-monotone}
	\If{$i(j) \neq \infty$}
		\For{$\ell = 1$ \textbf{to} $\lceil \eps^{-1} \rceil$}
			\State \parbox[t]{140mm}{Initialize an instance $ALG_{j, \ell}$ of {\FEA} that uses $g(T) = f(T \mid S_{j-1})$ as the objective function, and the matroids $M_1 / S_{j-1}, M_2 / S_{j-1}, \dotsc, M_k / S_{j-1}$ as the constraints.\strut}
		\EndFor
		\For{every element $u$ that arrives in the input stream} \label{line:elements_reading_non-monotone}
			\If{$f(u \mid S_{j - 1}) \in (\tau / (1 + \eps)^{i(j) + 1}, \tau / (1 + \eps)^{i(j)}]$}
				\For{$\ell = 1$ \textbf{to} $\lceil \eps^{-1} \rceil$} \label{line:copies_feeding}
					\State Feed $u$ to $ALG_{j, \ell}$.
					\If{$ALG_{j, \ell}$ marks $u$}
						\State \parbox[t]{122mm}{Guess whether $u = o_j$. If this is the case, denote $u$ by $u_j$ and exit the loop starting on \cref{line:elements_reading_non-monotone}.\strut} \label{line:o_guess}
						\State \parbox[t]{122mm}{Guess whether $u$ is the first element marked by $ALG_{j, \ell}$ such that $S_{j - 1} \cup \{o_{j + 1}, o_{j + 2}, \dotsc, o_{|OPT|}\} + u$ is a common independent set. If this is the case, denote $u$ by $u_{j, \ell}$.\strut} \label{line:u_j_choose_non-monotone}
						\State Exit the loop starting on \cref{line:copies_feeding}.
					\EndIf
				\EndFor
			\EndIf
			\If{some element was already denoted as $u_{j, \ell}$ for every integer $1 \leq \ell \leq \lceil \eps^{-1} \rceil$}
				\State Choose uniformly at random an integer $1 \leq \ell \leq \lceil \eps^{-1} \rceil$, and set $u_j \gets u_{j, \ell}$. \label{line:random_decision}
				\State Exit the loop starting on \cref{line:elements_reading_non-monotone}.
			\EndIf
		\EndFor
		\State Set $S_{j + 1} \gets S_{j} + u_{j}$.
	\Else
		\State Set $S_{j + 1} \gets S_{j}$.
	\EndIf
\EndFor
\State \Return{$S_{|OPT|}$}.
\end{algorithmic}
\end{algorithm}

The analysis of \cref{alg:non-monotone_SMkM} follows quite closely the analysis of \cref{alg:monotone_SMkM}. However, the non-monotonicity of the objective function and the extra logic within the algorithm itself makes the analysis of \cref{alg:non-monotone_SMkM} a bit more involved. The first part of this analysis is the following lemma, which is a counterpart of \cref{lem:u_k_selection} from the analysis of \cref{alg:monotone_SMkM}.

\begin{lemma} \label{lem:u_k_selection_non-monotone}
For every integer $1 \leq j \leq |OPT|$, the set $S_j \cup \{o_{j + 1}, o_{j + 2}, \dotsc, o_{|OPT|}\}$ is a common independent set. Furthermore, if $i(j) \neq \infty$, then \cref{alg:non-monotone_SMkM} selects some element as $u_j$, and this element appears weakly before $o_j$ in the input stream. Finally, if \cref{alg:non-monotone_SMkM} also selects elements $u_{j, 1}, u_{j, 2}, \dotsc, u_{j, \lceil \eps^{-1} \rceil}$, these elements also appear weakly before $o_j$.
\end{lemma}
\begin{proof}
Like in the proof of \cref{lem:u_k_selection}, we prove the current lemma by induction on $j$ with the induction hypothesis being that the lemma holds for all lower values of $j$. By the induction hypothesis, the set $S_{j-1} \cup \{o_{j}, o_{j + 1}, \dotsc, o_{|OPT|}\}$ is a common independent set. Furthermore, when $i(j) = \infty$, the lemma follows from this observation by the same argument used to prove the corresponding claim in the proof of \cref{lem:u_k_selection}. Therefore, we concentrate below on the case of $i(j) \neq \infty$. Finally, we note that if \cref{alg:non-monotone_SMkM} picks $o_j$ as $u_j$, then the lemma becomes an immediate corollary of the induction hypothesis because in this case $S_j \cup \{o_{j + 1}, o_{j + 2}, \dotsc, o_{|OPT|}\} = S_{j - 1} \cup \{o_{j}, o_{j + 1}, \dotsc, o_{|OPT|}\}$ and the $j$-th iteration of the algorithm terminates immediately after it reads $o_j$. Thus, we can safely assume below that $u_j \neq o_j$.

Let $u_{\text{first}}$ be the first element in the remaining data stream at the beginning of iteration number $j$ of \cref{alg:non-monotone_SMkM}. One can note that $u_{\text{first}}$ must be either (i) the first element of the original input stream, (ii) the element following $u_{j', \ell}$ for some integers $1 \leq j' < j$ and $1 \leq \ell \leq \lceil \eps^{-1} \rceil$ for which $i(j') \neq \infty$ or (iii) the element following $u_{j'}$ for some integer $1 \leq j' < j$ for which $i(j') \neq \infty$. In all these cases $o_j$ appears weakly after $u_{\text{first}}$ in the input stream because the induction hypothesis guarantees that $u_{j', \ell}$ and $u_{j'}$ appear weakly before $o_{j'}$ in the stream.

For every integer $1 \leq \ell \leq \lceil \eps^{-1} \rceil$, consider now an input stream $\cI_\ell$ that is a sub-stream of the original input stream.
Specifically, an element $u$ belongs to $\cI_\ell$ if (1) it appears weakly after $o_j$, or (2) it appears weakly after $u_{\text{first}}$ and is forwarded by \cref{alg:non-monotone_SMkM} to $ALG_{j, \ell}$.
By the induction hypothesis, $\{o_j, o_{j + 1}, \dotsc, o_{|OPT|}\}$ is a subset of the stream $\cI_\ell$ that is independent in all the matroids $M_1 / S_{j - 1}, M_2 / S_{j - 1}, \dotsc, M_k / S_{j - 1}$. Therefore, given the input stream $\cI_\ell$, $ALG_{j,\ell}$ is guaranteed to mark some element $u'_{j, \ell}$ that arrives weakly before $o_j$ and also obeys that
\[
	S_{j - 1} \cup \{o_j, o_{j + 1}, \dotsc, o_{|OPT|}\} - o_j + u'_{j, \ell}
	=
	S_{j - 1} \cup \{o_{j + 1}, \dotsc, o_{|OPT|}\} + u'_{j, \ell}
\]
is a common independent set. In reality, $ALG_{j, \ell}$ does not get the stream $\cI_\ell$. However, since $ALG_{j,\ell}$ makes its decisions online, it will mark $u'_{j, \ell}$ if given any prefix of $\cI_{j, \ell}$ that includes $u'_{j, \ell}$. Since $u'_{j, \ell}$ arrives weakly before $o_j$, this implies that $u'_{j, \ell}$ is marked by $ALG_{j, \ell}$ unless it is both equal to $o_j$ and marked by $ALG_{j, \ell'}$ for some $\ell' \neq \ell$. However, the last option cannot happen since \cref{alg:non-monotone_SMkM} sets $u_j$ to be $o_j$ given this option, and we assumed above that this does not happen. Given these arguments, \cref{alg:non-monotone_SMkM} must select either $u'_{j, \ell}$ or an earlier element as $u_{j, \ell}$ for every integer $1 \leq \ell \leq \lceil \eps^{-1} \rceil$, and therefore, $u_{j, \ell}$ weakly arrives before $o_j$.

As $u_j$ is equal to $u_{j, \ell}$ for some integer $1 \leq \ell \leq \lceil \eps^{-1} \rceil$, the conclusion of the last paragraph implies the lemma since
\[
	S_{j - 1} \cup \{o_{j + 1}, o_{j + 2}, \dotsc, o_{|OPT|}\} + u_j
	=
	S_j \cup \{o_{j + 1}, o_{j + 2}, \dotsc, o_{|OPT|}\}
\]
is a common independent set because $S_{j - 1} \cup \{o_{j + 1}, \dotsc, o_{|OPT|}\} + u_{j, \ell}$ is.
\end{proof}

At this point we need to present the following known lemma.
\begin{lemma}[Lemma~2.2 of~\cite{buchbinder2014submodular}] \label{lem:randomness_bound}
Given a non-negative submodular function $g\colon 2^\cN \to \nnR$ and a probability $p$, if $A(p)$ is a random subset of $\cN$ that includes every element with probability at most $p$ (not necessarily independently), then $\bE[g(A(p))] \geq (1 - p) \cdot g(\varnothing)$.
\end{lemma}

We are now ready to prove the approximation guarantee of \cref{alg:non-monotone_SMkM}.

\begin{corollary}
\cref{alg:non-monotone_SMkM} is a $(2 + O(\eps))$-approximation algorithm that outputs a set $S_{|OPT|}$ which is a common independent set.
\end{corollary}
\begin{proof}
The second part of the corollary follows immediately from \cref{lem:u_k_selection_non-monotone} by plugging in $j = |OPT|$. Therefore, we concentrate on proving the first part of the corollary.

Repeating the arguments used for this purpose in the proof of \cref{cor:approximation_monotone}, one can argue that for every integer $1 \leq j \leq |OPT|$ for which $o_j \not \in S_{|OPT|}$ we have
\begin{align} \label{eq:non_monotone_element_guarantee}
	f(S_j) - f(S_{j - 1})
	\geq{} &
	\frac{f(o_j \mid S_{j - 1} \cup \{o_{j + 1}, o_{j + 2}, \dotsc, o_{|OPT|}\})}{1 + \eps} - \frac{\eps \tau}{j^2}\\\nonumber
	\geq{} &
	\frac{f(o_j \mid S_{|OPT|} \cup \{o_{j + 1}, o_{j + 2}, \dotsc, o_{|OPT|}\})}{1 + \eps} - \frac{\eps \tau}{j^2}\\\nonumber
	={} &
	\frac{f(S_{OPT} \cup \{o_{j}, o_{j + 1}, \dotsc, o_{|OPT|}\})- f(S_{OPT} \cup \{o_{j + 1}, o_{j + 2}, \dotsc, o_{|OPT|}\})}{1 + \eps} - \frac{\eps \tau}{j^2}
	\enspace,
\end{align}
where the second inequality follows from the submodularity of $f$. Since $f(S_{OPT} \cup \{o_{j}, o_{j + 1}, \dotsc,\allowbreak o_{|OPT|}\})- f(S_{OPT} \cup \{o_{j + 1}, o_{j + 2}, \dotsc, o_{|OPT|}\})$ whenever $o_j \in S_{|OPT|}$ and \cref{alg:non-monotone_SMkM} never chooses an element $u_j$ for which $f(S_j) - f(S_{j - 1}) = f(u_j \mid S_{j - 1}) < 0$, adding up Inequality~\eqref{eq:non_monotone_element_guarantee} over all $j$ values yields
\begin{align} \label{eq:non-monotone-raw-bound}
	f(S_{|OPT|})
	\geq{} &
	f(S_0) + \frac{f(OPT \cup S_{|OPT|}) - f(S_{|OPT|})}{1 + \eps} - \sum_{j \mid u_j \not \in S_{|OPT|}} \mspace{-18mu} \frac{\eps \tau}{j^2}\\ \nonumber
	\geq{} &
	\frac{f(OPT \cup S_{|OPT|}) - f(S_{|OPT|})}{1 + \eps} - 6\eps \cdot f(OPT)
	\enspace,
\end{align}
where the second inequality holds by the non-negativity of $f$.

To make the last inequality useful, we need to prove an additional property of the set $S_{|OPT|}$. One can observe that in every iteration in which $i(j) \neq \infty$, \cref{alg:non-monotone_SMkM} either adds $o_j$ to its solution, or adds a uniformly random element out of the list $u_{j, 1}, u_{j, 2}, \dotsc, u_{j, \lceil \eps^{-1} \rceil}$. The elements in this list are disjoint since \cref{alg:non-monotone_SMkM} stops feeding an element $u$ to additional instances of {\FEA} once it is denoted as $u_{j, \ell}$ for some $\ell$, which prevents $u$ from being denoted again as another element in the above list. Furthermore, an element cannot belong to the list $u_{j, 1}, u_{j, 2}, \dotsc, u_{j, \lceil \eps^{-1} \rceil}$ for two different $j$ values because all the elements in this list are read by \cref{alg:non-monotone_SMkM} in its $j$-th iteration. These observations imply together that every element appears in $S_{|OPT|}$ with probability at most $1 / \lceil \eps^{-1} \rceil \leq \eps$, except for the elements of $OPT$ itself. Since $g(S) = f(S \cup OPT)$ is a non-negative submodular function, we now get by \cref{lem:randomness_bound} that
\[
	\bE[f(OPT \cup S_{|OPT|})]
	=
	\bE[g(S_{|OPT|} \setminus OPT)]
	\geq
	(1 - \eps) \cdot g(\varnothing)
	=
	(1 - \eps) \cdot f(OPT)
	\enspace.
\]

Plugging the last inequality into (the expectation of) Inequality~\eqref{eq:non-monotone-raw-bound} completes the proof of the lemma since $\eps \in (0, 1/7)$.
\end{proof}

To complete the proof of \cref{th:matroid_intersection_alg_simplified} for the case of an objective function that is not guaranteed to be monotone, it remains to analyze the space complexity of \cref{alg:non-monotone_SMkM}, which we do next.
\begin{observation} \label{obs:non-monotone_space_complexity}
\Cref{alg:non-monotone_SMkM} can be implemented so that its space complexity is
\[
	\tilde{O}(\eps^{-3\rho-1}k^{\rho k(2\rho + 1)}\rho(\log \rho + \log \eps^{-1})^{\rho + 1})
	=
	O\mleft(\mleft(\frac{k^{2\rho k}(\log \rho + \log \eps^{-1})}{\eps^3}\mright)^{\rho + 1}\rho\mright)
	\enspace.
\]
\end{observation}
\begin{proof}
We implement both the guesses and random decisions in \cref{alg:non-monotone_SMkM} by forking execution. In other words, if the guess or random decision has $Y$ possible outcomes, then the ``thread'' of the algorithm performing the guess or random decision splits into $Y$ parallel threads when the guess or random decision has to be done, and each thread executes with one of the possible outcomes. This guarantees that, if we restrict ourselves to the threads that made correct guesses, then the outputs of these threads are the support of the output distribution that a direct implementation of \cref{alg:non-monotone_SMkM} would have had. Hence, by outputting the best among all the feasible solutions produced by all the threads, we are guaranteed to output a solution which is as good as a the solution that would have been produced if we made all the right guesses and used true randomness for the random decisions (which is the solution whose approximation guarantee was analyzed above).

Let us now bound the number of threads used in the above implementation. \cref{alg:non-monotone_SMkM} makes two kinds of guesses/random decisions in each iteration. The first kind is done in \cref{line:guess_i_non-monotone} of the algorithm, which guesses a value for $i(j)$. The number of options for this value is
\[
	1 + \lceil \log_{1 + \eps} (j^2/\eps) \rceil
	\leq
	1 + \lceil \log_{1 + \eps} (\rho^2/\eps) \rceil
	=
	O(\eps^{-1}(\log \rho + \log \eps^{-1}))
	\enspace.
\]
The rest of the guesses and random decisions are used to pick as $u_j$ some element among the elements marked by the algorithms $ALG_{j, 1}, ALG_{j, 2}, \dotsc, ALG_{j, \lceil \eps^{-1} \rceil}$. One can replace all the guesses and random decisions of this kind by two guesses.
\begin{itemize}
	\item A guess of the index in the list of elements marked by $ALG_{j, 1}, ALG_{j, 2}, \dotsc, ALG_{j, \lceil \eps^{-1} \rceil}$ that points to the element that should become $u_j$.
	\item A guess of the index in the above list of the element that was the last element to be read by \cref{alg:non-monotone_SMkM} in the $j$-th iteration (notice that the algorithm can terminate iteration $j$ only after reading an element that is marked by $ALG_{j, \ell}$ for some $\ell$).
\end{itemize}
Since the output of $ALG_{j, \ell}$ is of size $O(k^{\rho k - 1})$ by \cref{thm:first_element}, the number of possible outcomes for the above two guesses is $O(\eps^{-2} k^{2\rho k - 2})$. Combining all the above observations, we get that the number of threads used by the \cref{alg:non-monotone_SMkM} increases by a factor of $O(\eps^{-3}k^{2\rho k - 2}(\log \rho + \log \eps^{-1}))$ following each iteration.

Since \cref{alg:non-monotone_SMkM} makes $|OPT| \leq \rho$ iterations, it has only $O(\eps^{-3\rho}k^{2\rho(\rho k - 1)}(\log \rho + \log \eps^{-1})^\rho)$ threads upon termination (which is the moment in which the number of threads is the largest). Each one of these threads requires enough space to store the index $i(j)$, $O(\rho + \eps^{-1})$ elements and the states of $ALG_{j, 1}, ALG_{j, 2}, \dotsc, ALG_{j, \lceil \eps^{-1} \rceil}$. By \cref{thm:first_element}, this can be done using a space complexity of
\[
	O(\log \eps^{-1} + \log \log \rho) + \tilde{O}(\rho + \eps^{-1}) + \tilde{O}(\eps^{-1}k^{k\rho + 1} \rho)
	=
	\tilde{O}(\eps^{-1} k^{k\rho + 1} \rho)
	\enspace.
\]
Multiplying the last bound with the bound we have on the number of threads used by \cref{alg:non-monotone_SMkM} completes the proof of the observation.
\end{proof}

\subsection{Algorithm for {\FESMkM}} \label{sec:first_element_SMkM}

In this section we prove \cref{thm:first_element}, which we repeat here for convenience.

\thmFirstElement*

The algorithm we use to prove \cref{thm:first_element} appears as \cref{alg:many_matroids}. In this algorithm we denote the elements of $\cN$ by $u_1, u_2, \dotsc, u_n$ in the order of their arrival. The algorithm maintains a set $S$ of states, which is updated during the processing of each arriving element. In the pseudocode of the algorithm, the set of states immediately after the processing of element $u_j$ is denoted by $S_j$. Each state $s \in S$ is a tuple of sets $s_1, s_2, \dotsc, s_k$, where the set $s_i$ is independent in the matroid $M_i$. We also use $\bar{\varnothing}$ to denote the tuple of $k$ empty sets, and use the shorthand $s(i \gets A)$ to denote the tuple obtained from $s$ by replacing $s_i$ with the set $A$.

Initially, the only state in the set $S$ of \cref{alg:many_matroids} is the tuple $\bar{\varnothing}$. Then, the algorithm starts to process the elements of the stream. Whenever an element $u_j$ arrives, the algorithm considers all the states $s \in S$. If there exists a set $s_i$ in $s$ to which $u_j$ cannot be added without violating the independence in the matroid corresponding to this set, then the state $s$ is simply kept in the set $S$ as is. Otherwise, the element $u_j$ is marked as an element of $U$ and the state $s$ is replaced with $k$ new states, each corresponding to the addition of $u$ to exactly one of the $k$ sets of the original state $s$.

\noindent \textbf{Remark.} We note that $u_j$ might be marked multiple times by \cref{alg:many_matroids} (potentially, once for every state $s \in S$). The element $u_j$ is considered to be marked as an element of $U$ if the algorithm marks it at least once, and is considered to be unmarked otherwise.

\begin{algorithm}
\caption{The Algorithm {\FEA} for {\FESMkM}} \label{alg:many_matroids}
\begin{algorithmic}[1]
\State Let $S_0 \gets \overline{\varnothing}$.
\For{every element $u_j \in \cN$ that arrives}
	\State Let $S_{j} \gets \varnothing$.
	\For{every state $s \in S_{j - 1}$}
		\If{$s_i + u_j \in \cI_i$ for all integer $1 \leq i \leq k$}\label{line:condition_multiple}
			\State Mark $u$ (as an element of $U$).
			\For{$i = 1$ \textbf{to} $k$}
				Add $s(i \gets (s_i + u_j))$ to $S_{j}$. \label{line:state_create}
			\EndFor
		\Else
			\State Add $s$ to $S_{j}$. 
		\EndIf
	\EndFor
\EndFor
\end{algorithmic}
\end{algorithm}

We begin the analysis of \cref{alg:many_matroids} by bounding the number of elements that it marks and its space complexity.

\begin{observation}
\cref{alg:many_matroids} marks $O(k^{\rho k - 1})$ elements, and its space complexity is $\tilde{O}(k^{\rho k + 1}\rho)$.
\end{observation}
\begin{proof}
Given a state $s$, let us define its size as $|s| = \sum_{i = 1}^k |s_i|$; and consider the potential function
\[
	\phi(S)
	=
	\sum_{s \in S} \frac{k^{\rho k - |s|} - 1}{k - 1}
	\enspace.
\]
Fix now some integer $1 \leq j \leq n$, and let us lower bound the difference $\phi(S_{j - 1}) - \phi(S_j)$. If $u_j$ is not marked by \cref{alg:many_matroids}, then the sets $S_{j - 1}$ and $S_j$ are identical, and thus, $\phi(S_{j - 1}) = \phi(S_j)$. Otherwise, we denote by $R$ the set of states in $S_{j - 1}$ for which the condition on \cref{line:condition_multiple} of \cref{alg:many_matroids} evaluated to FALSE. Since $u_j$ is marked, $R$ is a strict subset of $S_{j - 1}$, and therefore,
\begin{align*}
	\phi(S_{j - 1})
	={} &
	\sum_{s \in S_{j - 1}} \frac{k^{\rho k - |s|} - 1}{k - 1}
	=
	\sum_{s \in R} \frac{k^{\rho k - |s|} - 1}{k - 1} + \sum_{s \in S_{j - 1} \setminus R} \frac{k^{\rho k - |s|} - 1}{k - 1}\\
	={} &
	\sum_{s \in R} \frac{k^{\rho k - |s|} - 1}{k - 1} + \sum_{s \in S_{j - 1} \setminus R} \left[1 + \sum_{i = 1}^k \frac{k^{\rho k - |s(i \gets (s + u_j))|} - 1}{k - 1}\right]\\
	\geq{} &
	\sum_{s \in S_j} \frac{k^{\rho k - |s|} - 1}{k - 1} + |S_{j - 1} \setminus R|
	\geq
	\sum_{s \in S_j} \frac{k^{\rho k - |s|} - 1}{k - 1} + 1
	\enspace.
\end{align*}
In conclusion, we have proved that the potential $\phi$ decreases by at least $1$ when $u_j$ is marked, and remains unchanged otherwise. Since the potential is always non-negative (observe that every set $s_i$ of a state $s$ must be independent in $M_i$, and thus, of size at most $\rho$), this implies that the number of elements marked by \cref{alg:many_matroids} is upper bounded by
\[
	\phi(S_0)
	=
	\phi(\bar{\varnothing})
	=
	\frac{k^{\rho k - |\bar{\varnothing}|} - 1}{k - 1}
	=
	\frac{k^{\rho k} - 1}{k - 1}
	=
	O(k^{\rho k - 1})
	\enspace.
\]

To prove the second part of the observation, we need to note that the above analysis in fact shows that the potential $\phi$ decreases by at least $|S_{j - 1} \setminus R|$, where $|S_{j - 1} \setminus R|$ is the number of states that are not copied as is from $S_{j - 1}$ to $S_j$. As each one of these states yields $k$ states in $S_j$, we get that the number of states that the algorithm maintains at every given point is at most $1$ plus $k$ times the total decrease in the potential $\phi$. Since this potential never increases, this implies that the algorithm keeps at most $1 + k \cdot \phi(S_0) = O(k^{\rho k})$ states. The second part of the observation now follows since each state consists of $k$ sets, each independent in at least one of the matroids $M_1, M_2, \dotsc, M_k$, and therefore, can be represented using $\tilde{O}(k \rho)$ space.
\end{proof}

Fix now an arbitrary common independent set $O \subseteq \cN$, and let $o_1$ be the first element of $O$ that arrives in the input stream. To prove that \cref{alg:many_matroids} is a valid algorithm for {\FESMkM}, we have to show that it always marks an element $u$ that arrives weakly before $o_1$ and obeys that $O - o_1 + u$ is also a common independent set. Towards this goal, let us say that a state $s$ \emph{has potential} if for every integer $1 \leq i \leq k$ and element $v \in s_i$ we have $O - o_1 + v \not \in \cI_i$. Clearly, the initial state $\bar{\varnothing}$ has potential (since its sets include no elements). The following lemma shows that we keep having a state with potential as long as we do not mark an element $u$ with the required properties.
\begin{lemma}
Assume there is a state $s \in S_{j - 1}$ which has potential for some integer $1 \leq j \leq n$, and $u_j$ is either not marked by \cref{alg:many_matroids} or obeys $O - o_1  + u_j \not \in \cI_i$ for at least one matroid $M_i$, then $S_j$ also includes a state that has potential.
\end{lemma}
\begin{proof}
If the state $s$ is copied to $S_j$, then we are done. Therefore, we can assume from now on that $s$ is not copied to $S_j$, which implies that $u_j$ is marked, and by the assumption of the lemma obeys $O - o_1 + u_j  \not \in \cI_i$ for at least one matroid $M_i$. One can observe that the last property implies that the state $s' = s(i \gets (s_i + u_j))$ has potential because the state $s$ has potential and the only difference between $s'$ and $s$ is the addition of the element $u_j$ to $s_i$. The lemma now follows since the state $s'$ belongs to $S_j$.
\end{proof}

The last lemma implies that either some element obeying the required properties is marked before the arrival of $o_1$, or there is a state with potential immediately before the arrival of $o_1$. The following lemma shows that $o_1$ is marked whenever this happens.

\begin{lemma}
Let us denote $u_{j'} = o_1$. If there exists a state $s \in S_{j' - 1}$ that has potential, then $o_1$ is marked by \cref{alg:many_matroids}.
\end{lemma}
\begin{proof}
Assume towards a contradiction that $o_1$ is not marked by \cref{alg:many_matroids}. That means that there exists some integer $1 \leq i \leq k$ such that $s_i$ spans $o_1$ in $M_i$. Thus, if we denote the rank function of $M_i$ by $\rank_i$, then, by the submodularity of $\rank_i$,
\begin{align*}
	\rank_i(O \cup s_i - o_1)
	={} &
	\rank_i(O \cup s_i) - [\rank_i(O \cup s_i) - \rank_i(O \cup s_i - o_1)]\\
	\geq{} &
	\rank_i(O \cup s_i) - [\rank_i(s_i + o_1) - \rank_i(s_i)]
	=
	\rank_i(O \cup s_i)
	\geq
	\rank_i(O)
	\enspace.
\end{align*}
This inequality implies that there exist an element $u \in s_i$ such that $O - o_1 + u \in \cI_i$, which contradicts our assumption that $s$ has potential.
\end{proof}

By the above discussion, the above lemmata imply the following corollary, which completes the proof of \cref{thm:first_element}.
\begin{corollary}
If no element $u$ that obeys $O - o_1 + u \in \cI_i$ for every integer $1 \leq i \leq k$ is marked by \cref{alg:many_matroids} before the arrival of $o_1$, then $o_1$ itself is selected by the algorithm.
\end{corollary}

\newcommand{\algsub}{\ensuremath{\mathcal{A}}}
\newcommand{\algmat}{\ensuremath{\mathcal{B}}}
\section{Impossibility Results for Bipartite Matching Constraints} 
\label{sec:main-hardness}

In this section we present our impossibility results for the problem of \texttt{Monotone Submodular Maximization subject to Bipartite Matching} (\MSMBM). In this problem the ground set consists of the set of edges of some bipartite graph $G = (V, E)$, and we are given a non-negative monotone submodular function $f\colon 2^\cN \to \nnR$. The objective is to output a feasible matching of $G$ maximizing $f$. In the presence of parallel edges, we improve the recent hardness of $1.916$ for this problem (that relies on complexity theoretic assumptions) to the following unconditional result (notice that this result is a restatement of \cref{thm:bipartite_matching}).
\thmBipartiteMatching*
This result is obtained by combining two hardness results: the one-way communication complexity of \chainPn{} and streaming lower bounds for the bipartite maximum matching problem. 
It is a longstanding open question whether it is possible to devise a data stream algorithm for the maximum matching problem with a better approximation guarantee than $2$, even if we are allowed to use memory $O(r^{2-\eps})$ for some constant $\eps > 0$. The following result basically says that improving over the guarantee $3$ for maximizing a monotone submodular function  subject to a bipartite matching constraint in the data stream model would lead to such a breakthrough. 

\begin{restatable}{theorem}{thmmainmatchingredreswithass}
For any constant $0 < \eps < 1$, assuming \cref{thm:matching-hard-nice} for $\alpha=2-\eps$ and $\mathcal{M} (r) = O(r^{2-\eps})$, any single-pass data stream algorithm for  {\MSMBM} on $r$-vertex graphs 
that finds a $(3-O(\eps))$-approximate solution with probability at least $2/3$ uses memory $\Omega (r^{2-\eps}/\log r)$.
\label{thm:hardness_matching_with_assumption}
\end{restatable}


\subsection{Preliminaries}

Our impossibility results harness the hardness of two problems:  bipartite matching in the data stream model and \chainPn{}.

\paragraph{Bipartite matching in the data stream model.} The task of devising optimal data stream algorithms for finding a matching in  (bipartite) graphs remains a notorious open question. 
In this problem, the algorithm is provided with a stream of edges of a bipartite graph $G = (V, E)$, and is allowed to use a limited amount of memory while processing them. 
We assume that before it starts to read the stream, the algorithm has unbounded computational power (i.e., unbounded time and memory) to  initialize, and similarly, after reading the last edge from the stream, the algorithm again has unbounded computational power to produce a matching based on what it has stored in its memory.  
Kapralov~\cite{DBLP:conf/soda/Kapralov21} proved that any (potentially randomized) $1.692$-approximate data stream algorithm for this problem with success probability at least $1/2$ requires $\Omega(|V| \log^{\omega(1)} |V|)$ memory. His result is information theoretic and he shows that there is a hard distribution of $r$-vertex bipartite graphs (for large enough $r$) so that no algorithm of ``small'' memory can find a ``good'' solution,  even if it has unbounded computational power while processing edges (but limited memory in-between the arrival of edges). In particular, his result implies the following theorem with $\mathcal{M}(r) = O(r \log^{O(1)} r)$ and $\alpha = 1.692$.

\begin{theorem} \label{thm:matching-hard-nice}
There is an infinite number of positive integers $r$  such that the following holds. 
Consider a single-pass data stream algorithm $\algmat$ for the bipartite matching problem that uses memory at most $\mathcal{M}(r)$. Then,  there  is  an $r$-vertex instance $G$   such that $\algmat$ finds an $\alpha$-approximate matching with probability at most $1/2$ on input $G$.
\end{theorem}
 We have stated the theorem in this general form as we will use it as a template in our general reduction. Indeed, it is a conceivable that the theorem holds with $\mathcal{M}(r)= O(r^{2-\varepsilon})$ and $\alpha = 2 - \varepsilon$, which leads to our stronger (conditional) lower bound \cref{thm:hardness_matching_with_assumption}. 
 Throughout, we assume that the memory satisfies $ x \leq \mathcal{M}(x) \leq x^2$ for large enough $x \in \bR$, and the approximation guarantee satisfies $ 0 < \alpha < 1$.
Finally, we have the following corollary obtained by running $q$ independent copies of $\algmat$   and outputting the largest found matching among all copies.
\begin{corollary}
Let $q\geq 1$ be an integer. 
There is an infinite number of positive integers $r$  such that the following holds. 
Consider a single-pass data stream algorithm $\algmat$ for the bipartite matching problem that uses memory at most $\mathcal{M}(r)/q$. Then,  there  is  a $r$-vertex instance $G$   such that $\algmat$ finds an $\alpha$-approximate matching with probability at most $1/q$ on input $G$.
%
%
\label{corollary:matching-hard-nice}
\end{corollary}

\subsection{Hardness reduction for bipartite matching constraint}

As mentioned above, we describe a general reduction that harnesses the hardness of the bipartite matching problem in the data stream model. This general reduction formally appears as \cref{thm:main-matching-reduction}. We note that the general reduction implies \cref{thm:bipartite_matching,thm:hardness_matching_with_assumption} by selecting $\eps$ to be small enough and $p$ to be large enough. Specifically, \cref{thm:hardness_matching_with_assumption} follows by substituting in the assumptions $\mathcal{M}(r) = O(r^{2-\eps})$ and $\alpha = 2 - \eps$,  and \cref{thm:bipartite_matching} is implied  since~\cite{DBLP:conf/soda/Kapralov21} proved \cref{thm:matching-hard-nice} with $\mathcal{M}(r) = O(r\log^c r)$ for any constant $c>1$ and  $\alpha = 1.692$.

\begin{theorem} \label{thm:main-matching-reduction}
Assuming \cref{thm:matching-hard-nice}, for any $\varepsilon>0$ and integer $p\geq 2$, any data stream algorithm for {\MSMBM}  that finds a $\frac{p}{1+\varepsilon} \left(\frac{\alpha+1}{p+\alpha+1}  \right) $-approximate solution with probability at least $2/3$ must use at least $\frac{\mathcal{M} (r)}{5 \cdot 10p\cdot \log_{1+\eps}(r)}$ memory, where $r$ denotes the number of vertices of the bipartite graph.
\end{theorem}

Let $r\gg p$   be a large integer as is guaranteed by~\cref{corollary:matching-hard-nice}, i.e., for any data stream algorithm $\algmat$ for the bipartite matching problem that  uses memory at most $\mathcal{M}(r)/q$, there is an $r$-vertex instance $G$ such that $\algmat$ finds an $\alpha$-approximate matching with probability at most $1/q$ on input $G$. Here, we select $q = 10p$. We further let $n = \mathcal{M}(r)/(5 \cdot 10p)$, and we assume that $r$ and $n$ are selected to be large enough so that  $O(p \log n) < n$ (for the hidden constants appearing in the proofs) and $n/\log_{1+\eps}(r) + p \log(n) < n/(36p^2)$. This allows us to simplify some (technical) calculations.

Our approach is to assume the existence of an algorithm $\algsub$  for {\MSMBM} that finds a  $\frac{p}{1+\varepsilon} \left(\frac{\alpha+1}{p+\alpha+1}  \right) $-approximate solution with probability at least $2/3$ on any instance.  Using this algorithm, we provide a protocol for {\chainPpn}. Our protocol is parameterized by $r$-vertex instances $G_1 = (V_1, E_1), \ldots, G_p = (V_p, E_p)$ to the bipartite matching problem in the streaming  model. These instances will later be selected to be ``hard'' instances using \cref{corollary:matching-hard-nice} (see \cref{sec:graph_selection}). Throughout, for $i \in [p]$, we use $m_i$  to denote the smallest power of $1+\varepsilon$ that upper bounds the size of a maximum matching in $G_i$. 

Each player in our protocol for \chainPpn will simulate $\algsub$ on a monotone submodular function selected from a certain family. We describe this family of submodular functions next. We then, in \cref{sec:protocol_description}, describe and analyze the protocol assuming a ``good'' selection of the instances $G_1, \ldots, G_p$. Finally, in \cref{sec:graph_selection}, we show how to select such instances and explain how it implies \cref{thm:main-matching-reduction}. 



\subsubsection{Family of submodular functions} 
\label{sec:family_submodular_functions}
We start by defining an extended ground set based on the edge sets of the graphs $G_1, \ldots, G_p$. For $i\in [p]$, let 
\[
    \cN_i = \{(e, j)\mid e \in E_i, j\in [n]\}\,.
\]
In other words, $\cN_i$ contains $n$ parallel copies of each edge of $G_i$, one for each possible choice of $j\in [n]$. We shall use the notation $\cN_{\leq i} = \cN_1 \cup \cdots \cup \cN_i$, $\cN_{ \geq i} = \cN_i \cup \cN_{i+1} \cup \cdots \cup \cN_p$, and $\cN = \cN_{\leq p}$. Furthermore,  for a subset $S \subseteq \cN$, we let
\[
    s(i) = \frac{|S \cap \cN_i|}{m_i} \qquad \mbox{and} \qquad {s(i,\neg o_i)} = \frac{|\{(e, j) \in S \cap \cN_i \mid j \neq o_i \}|}{m_i}\enspace.
\]
Recall that $m_i$ denotes  (an upper bound on) the size of a maximum matching in $G_i$. Hence, assuming the edges in $S \cap \cN_i$ form a matching in $G_i$, $s(i) \in [0,1]$  denotes the approximation ratio of the considered matching. Similarly, ${s(i, \neg o_i)}$ measures the approximation ratio of those edges that  do \emph{not} correspond to some index $o_i$.

We now recursively define $p$ families of non-negative monotone submodular functions $\cF_p, \cF_{p-1},\allowbreak \ldots, \cF_1$.  Family $\cF_p$ contains a single monotone submodular function $f_{o_{p}}\colon 2^{\cN_p} \rightarrow \nnR$ defined by $f_{o_p}(S) = \min\left\{1, s(p)\right\}$. The use of the subindex $o_{p}$ in the last definition is not technically necessary, but it simplifies our notation.  For $i=p-1, \ldots, 1$, the family $\cF_i = \{f_{o_i,  \ldots, o_p} \mid o_i, \ldots, o_{p-1} \in [n]\}$ consists of $n^{p-i}$ monotone submodular functions on the ground set $\cN_{\geq i}$ that are defined recursively in terms of the functions  in $\cF_{i+1}$ as follows:
\begin{equation}
    f_{o_{i}, \ldots,  o_{p}}(S) =  \min\left\{p+1-i, {s(i)} + \left( 1 - \frac{s(i, \neg o_i)}{p+1-i}\right) f_{o_{i+1}, \ldots, o_{p}} (S \cap \cN_{\geq i+1}) \right\}\enspace .
    \label{eq:def_of_f}
\end{equation}
To intuitively understand the last definition. One should think of every graph $G_i$ as having a mass to be covered (in some sense). Every edge of $G_i$ covers equal amounts of mass from $G_i, G_{i + 1}, \dotsc, G_p$, except for the edges of $G_i$ with the index $o_i$, which cover only mass of $G_i$. Furthermore, edges of a single graph $G_i$ are correlated in the sense that the mass of $G_{i'}$ that they cover (for any $i \leq i' \leq p$) is additive, while edges of different graphs $G_i$, $G_{i'}$ that cover the mass of the same graph $G_{i''}$ do it in an independent way (so if the edges of $G_i$ cover a $q_1$ fraction of this mass and the edges of $G_{i'}$ cover $q_2$ fraction, then together they cover only $1 - (1 - q_1)(1 - q_2)$ of the mass of $G_{i''}$). Given this intuitive point of view, $f_{o_i, o_{i + 1}, \dotsc, o_p}(S)$ represents the total mass of the graphs $G_i, G_{i + 1}, \dotsc, G_p$ that is covered by the edges of $S$. Note that this explains why \cref{eq:def_of_f} includes a negative term involving $s(i, \neg o_i)$: if there are many edges of $G_i$ with indexes other than $o_i$ that appear in $S$, then a lot of the mass accounted for by $f_{o_{i + 1}, o_{i + 2}, \dotsc, o_p}(S)$ is counted also by $s(i)$.

%
\begin{observation}
For every $i \in [p]$, the functions of $\cF_i$ are non-negative, monotone and submodular.
\end{observation}
\begin{proof}
We prove the observation by downward induction on $i$. For $i = p$ it follows because $f_{o_p}$ is the minimum between a positive constant and the non-negative monotone and submodular function $s(i)$, and such a minimum is known to also have these properties (see, e.g., Lemma 1.2 of~\cite{buchbinder2018submodular}).

Assume now that the observation holds for $i + 1$, and let us prove it for $i$. The product
\[
	\left( 1 - \frac{s(i, \neg o_i)}{p+1-i}\right) f_{o_{i+1}, \ldots, o_{p}} (S \cap \cN_{\geq i+1})
\]
is the product of two non-negative submodular functions, one of which is monotone and the other down-monotone (i.e., $f(S) \geq f(T)$ for every $S \subseteq T$), and such products are known to be non-negative and submodular.\footnote{To see why, note that if $f$ is a non-negative monotone submodular function and $g$ is a non-negative down-monotone submodular function, then, with respect to the product $f \cdot g$, the marginal contribution of an element $u$ to a set $S$ that does not include it is given by $f(S + u) \cdot g(S + u) - f(S) \cdot g(S) = f(u \mid S) \cdot g(S + u) + f(S) \cdot g(u \mid S)$, which is a down-monotone function of $S$.} Therefore, the sum
\[
	s(i) + \left( 1 - \frac{s(i, \neg o_i)}{p+1-i}\right) f_{o_{i+1}, \ldots, o_{p}} (S \cap \cN_{\geq i+1})
\]
is non-negative and submodular since the sum of non-negative and submodular functions also has these properties (see, again, Lemma 1.2 of~\cite{buchbinder2018submodular}). This sum is also monotone since adding an edge $e \in \cN_i$ to $S$ either increases the sum by $1$, if the index of the edge $e$ is $o_i$, or by $1 - f_{o_{i + 1}, o_{i + 2}, \dotsc, o_p} / (p + 1 - i) \geq 0$, otherwise; and adding an edge $e \in \cM_{\geq i + 1}$ to $S$ can only increase the sum because $f_{o_{i+1}, \ldots, o_{p}}$ is monotone and $s(i, \neg o_i) \leq 1 \leq p + 1 - i$. The observation now follows since $f_{o_{i}, \ldots,  o_{p}}(S)$ is the minimum between the above sum and a positive constant.
\end{proof}

We let $\cF = \cF_1$. The following lemma proves some useful properties of this functions family.
\crefname{enumi}{Property}{Properties}
\begin{lemma}
   The monotone submodular functions in $\cF$ have the following properties: 
    \begin{enumerate}[label=(\alph*)]
            \item  For $i\in [p]$, any two functions $f_{o_1, \ldots, o_p}, f_{o'_1, \ldots, o'_p} \in \cF$ with $o_1 = o'_1, \ldots, o_{i-1} = o'_{i-1}$ are identical when restricted to the ground set $\cN_{\leq i}$. \label{property:indistinguishable}
        \item  We can evaluate $f_{o_1, \ldots, o_p}$, on input set $S$, using $n$ memory in addition to the input length.  \label{property:low_memory}
        \item Let $M_1, \ldots, M_p$ be maximum matchings in $G_1, \ldots, G_p$, respectively. Then
        \[
            f_{o_1, \ldots, o_p} (S) \geq p/ (1+\varepsilon) \qquad \mbox{for } S = \bigcup_{i=1}^p \{(e, o_i) \mid e\in M_i\}\enspace . 
        \] \label{property:yescase}
        \item For a subset $S \subseteq \cN$ such that $|S \cap \cN_i| \leq m_i / \alpha$ and  $\{(e, o_i) \in S \cap \cN_i\} = \varnothing$ for all $i\in [p]$,
        \[
            f_{o_1, \ldots, o_p} (S) < 1 + \frac{1}{\alpha+1}p\enspace .
        \] \label{property:nocase}
    \end{enumerate}
    \label{lemma:gen_hardness_F_properties}
\end{lemma}

\begin{proof}
\emph{\cref{property:indistinguishable}:} This follows by the definition of the submodular functions~\eqref{eq:def_of_f}: when  $S \subseteq \cN_{\leq i}$, the value of $f_{o_1, \ldots, o_p}(S)$ only depends on $s(1) , s(1, \neg o_1), \ldots, s(i-1), s(i-1, \neg o_{i-1})$ and $s(i)$.

\bigskip

\emph{\cref{property:low_memory}:} Given input set $S$, we can evaluate $f_{o_1, \ldots, o_p}(S)$ as follows. First calculate $s(p)$ and $s(p, \neg o_p)$. This requires us to store two numbers. Furthermore, using these numbers we calculate $f_{o_p}(S \cap \cN_p)$ and then ``free'' the memory used for $s(p)$ and $s(p, \neg o_p)$. 
Now, suppose we have calculated $f_{o_{i+1}, \ldots, o_p}(S \cap \cN_{\geq i+1})$. We then calculate $s(i)$ and $s(i, \neg o_i)$ which allows us to calculate $f_{o_i, \ldots, o_p}(S \cap \cN_{\geq i})$ from $f_{o_{i+1}, \ldots, o_p}(S \cap \cN_{\geq i+1})$. 
Following this calculation, we free the memory used for $f_{o_{i+1}, \ldots, o_p}(S \cap \cN_{\geq i+1})$ and $s(i+1), s(i+1, \neg o_{i+1})$. 
We proceed in this way until we have calculated the desired value $f_{o_1, \ldots, o_p}(S)$. At any point of time we have only stored at most $4$ numbers, and each number takes $O(\log n)$ bits to store. Thus the memory that we need is upper bounded by $n$ (since $n$ is selected to be sufficiently large).

\bigskip

\emph{\cref{property:yescase}} For $S=  \bigcup_{i=1}^p \{(e, o_i) \mid e\in M_i\}$, by the selection of $M_i$ and $m_i$, $s(i) \geq 1/(1+\varepsilon)$ for all $i\in [p]$. Moreover, since we only have items corresponding to the indices $o_1, \ldots, o_p$, we have $s(i, \neg o_i) = 0$ for $i\in [p]$. Hence,  by~\eqref{eq:def_of_f}, 
\(
    f_{o_1, \ldots, o_p}(S)   \geq p/(1+\varepsilon).
\)

\bigskip

\emph{\cref{property:nocase}:} As $S$ does not contain any elements with the indices $o_1, \ldots, o_p$, we have $s(i) = s(i, \neg o_i)$ for all $i\in [p]$. Furthermore, by assumption, $s(i) \leq 1/\alpha <1$,
which implies $f_{o_p}(S \cap \cN_{p}) = s(p)$ and, for $i=p-1, \ldots, 1$,
\[
    f_{o_{i}, \ldots, o_p}(S \cap \cN_{\geq i}) = s(i)  + \left(1-\frac{s(i)}{p+1-i}\right) f_{o_{i+1}, \ldots, o_p}(S \cap \cN_{\geq i+1})\enspace. 
\]
In the last equality we only used the fact that $s(i) < 1$ for all $i \in [p]$. Plugging in the stronger inequality $s(i) \leq 1/\alpha$ yields
\begin{align*}
f_{o_{1}, \ldots, o_p}(S ) 
\leq {1/\alpha } &\left( 1+\left(1-\frac{1/\alpha}{p}\right) + \left(1-\frac{1/\alpha}{p}\right)\left(1-\frac{1/\alpha}{p-1}\right)\right. \\ & \quad \quad +\left(1-\frac{1/\alpha}{p}\right)\left(1-\frac{1/\alpha}{p-1}\right)\left(1-\frac{1/\alpha}{p-2}\right)   \\
 & \quad \quad +  \ldots + \left. \left(1-\frac{1/\alpha}{p}\right)\left(1-\frac{1/\alpha}{p-1}\right) \ldots \left(1-\frac{1/\alpha}{1}\right)\right) \enspace.
\end{align*}

Let us find the solution to the terms inside the parenthesis in the last inequality. To that end, we let  
\begin{align*} h(p) \mspace{-1mu}=\mspace{-1mu} \mleft(1\mspace{-1mu}+\mspace{-1mu}\mleft(1-\frac{1/\alpha}{p}\mright) \mspace{-1mu}+\mspace{-1mu} \mleft(1-\frac{1/\alpha}{p}\mright)\mspace{-1mu}\mleft(1-\frac{1/\alpha}{p-1}\mright)\mspace{-1mu}+\mspace{-1mu}\mleft(1-\frac{1/\alpha}{p}\mright)\mspace{-1mu}\mleft(1-\frac{1/\alpha}{p-1}\mright)\mspace{-1mu}\mleft(1-\frac{1/\alpha}{p-2}\mright) \mspace{-1mu}+\mspace{-1mu} \ldots \mright) .
\end{align*}
One can observe that
\[h(p) = (1-\frac{1/\alpha}{p})h(p-1)+1 \quad \quad \text{ and } \quad\quad h(0) = 1.\]
Let us now show by induction that $h(k) \leq \frac{k\alpha+\alpha+1}{\alpha+1}$, and the inequality is strict for every $k \geq 1$. The base case $h(0) = 1$ holds as an equality. Now, suppose we have proved $h(k-1) \leq \frac{k\alpha+1}{ \alpha+1}$, then

\begin{align*}
h(k) ={}& \left(1-\frac{1/\alpha}{k}\right) h(k-1) +1 \\
     \leq{}& \left(1-\frac{1/\alpha}{k}\right) \cdot \left(\frac{k\alpha+1}{\alpha+1}\right) + 1 \\
     ={}& \left(\frac{k\alpha-1}{k\alpha}\right) \cdot \left(\frac{k\alpha+1}{\alpha+1}\right) + 1 \\
     ={}& \frac{k^2\alpha^2+k\alpha^2+k\alpha-1}{k\alpha(\alpha+1)}\\
     <{}& \frac{k\alpha(k\alpha+\alpha+1)}{k\alpha(\alpha+1)} \\
     ={}& \frac{k\alpha+\alpha+1}{\alpha+1}\, ,
\end{align*}
where the first inequality holds since $1 < \alpha \leq k$. 

Plugging the upper bound we have proved on $h(p)$ into the upper bound we have on $f_{o_1, \ldots, o_p}(S)$ produces
\[
    f_{o_1, \ldots, o_p}(S) \leq 1/\alpha \cdot h(p) \leq 1/ \alpha \cdot \frac{p\alpha+\alpha + 1}{\alpha + 1} < 1 + \frac{1}{\alpha+1}p 
		\enspace,
\]
as required.
\end{proof}

\subsubsection{Description and analysis of protocol for \texorpdfstring{\chainPpn}{chain index problem}}
\label{sec:protocol_description}

We use algorithm \algsub{} for {\MSMBM} to devise \cref{alg:chainreduction} for \chainPpn.  Given \chainPpn instance $x^1, x^2, t^2, \ldots, x^p, t^p, t^{p+1}$, the protocol simulates the execution of $\algsub$ on the following stream:
\begin{quote}
    First the elements  in $\{(e,j) \in \cN_1 \mid j\in [n]\mbox{ with } x^1_j = 1\}$ are given (by the first player). Then for, $i=2, \ldots, p$, the elements in 
    $\{(e,j) \in \cN_i \mid j\in [n] \mbox{ with } x^i_j = 1\}$ are given (by the $i$-th player).
\end{quote}
The submodular function to be optimized  is $f_{o_1, \ldots, o_p}$ where $o_i = {t^{i+1}}$ for $i=1, \ldots, p-1$. In order for the players to be able to simulate the execution of \algsub{} and any oracle call made to $f_{o_1, \ldots, o_p}$, Player $i$ sends to Player $i+1$ the state of $\algsub$ and the indices $t^2, \ldots, t^i$. Hence, the communication complexity of the protocol is upper bounded by the memory usage of $\algsub$ plus $p \log n$. Note that Player $i$ only needs to know $o_1, \ldots, o_{i-1}$ (and thus indices $t^2,\ldots, t^i$) in order to evaluate the oracle calls by \cref{property:indistinguishable} of \cref{lemma:gen_hardness_F_properties} (since at that point only elements of $\cN_{\leq i}$ has arrived, and hence, \algsub{} can only query the oracle for subsets of $\cN_{\leq i}$).

\begin{protocol}[t]
\caption{Reduction from \chainPpn to {\MSMBM} } \label{alg:chainreduction}

\textbf{Player $P_i$'s Algorithm for $i=1, \ldots, p$}
\begin{algorithmic}[1]
    \State Initialize $\algsub$ with the received memory state  (or initial state if first player).
    \State Simulate $\algsub$ on the elements  $\{(e,j) \in \cN_i  \mid j\in [n] \text{ with }x^{i}_j = 1\}$.
    \State The objective function for $\algsub$ is
    one  of the functions $f_{o_1, \ldots, o_p} \in \cF$ with 
        $o_1 = t^2, \ldots, o_{i-1} = t^i$. 
    By \cref{property:indistinguishable} of \cref{lemma:gen_hardness_F_properties}, these functions are identical when restricted to $\cN_{\leq i}$, and so any oracle query from $\algsub$ can be evaluated without ambiguity. 
    \State Send to $P_{i+1}$ the values $t^2, t^3, \ldots, t^i$ and the memory state of \algsub.
\end{algorithmic}
\textbf{Player $P_{p+1}$'s Algorithm}
\begin{algorithmic}[1]
    \State The objective function for $\algsub$ can now be determined to be  
    $f_{o_1, \ldots, o_p} \in \cF$ with 
        $o_1 = t^2, \ldots, o_{p-1} = t^p$. 
     \State Initialize $\algsub$ with the received memory state,  and ask it to return a solution $S$.
    \State If  $f_{o_1, \ldots, o_p}(S) \geq 1 + \frac{1}{1+\alpha} p $, output ``$1$-case''; otherwise, output ``$0$-case''.
\end{algorithmic}
\end{protocol}

We proceed to analyze the success probability of the protocol. The success probability in the $0$-case will depend on the selection of $G_1, \ldots, G_p$. 

\begin{definition}
    We say that the selection of $G_1, \ldots, G_p$ is \emph{successful} if the following holds: if we select a random $0$-case instance of \chainPpn{} from $D(p+1, n)$, then with probability at least $9/10$ the output $S$ of \algsub{} in \cref{alg:chainreduction} satisfies $|S \cap \cN_i| \leq  m_i / \alpha$ for all $i\in[p]$.
\end{definition}
In other words,  the selection of $G_1, \ldots, G_p$ is successful if   \cref{alg:chainreduction} is unlikely to find a large matching in any of the graphs. Intuitively, it should be possible to select such graphs since, by \cref{thm:hardness_matching_with_assumption}, any algorithm for finding a large matching requires large memory. The following lemma formalizes this argument.
\begin{restatable}{lemma}{successfullemma}
    If $\algsub$ uses memory at most $n/\log_{1+\epsilon}(r)$, there is a successful selection of $G_1, \ldots, G_p$.%
    \label{lemma:successful_selection}
\end{restatable}

The next section is devoted to proving the last lemma. Here we proceed to show how it implies \cref{thm:main-matching-reduction}.


\begin{lemma}
    If the selection of $G_1, \ldots, G_p$ is successful,   \cref{alg:chainreduction} succeeds with probability at least $2/3$ on the distribution $D(p+1,n)$. %
    \label{lemma:analysis_if_successful}
\end{lemma}
Below we prove \cref{lemma:analysis_if_successful}. However, before doing so, let us first explain how \cref{lemma:analysis_if_successful} implies \cref{thm:main-matching-reduction}. Indeed, suppose toward contradiction that \algsub{} uses memory less than $n/\log_{1+\eps}(r)$ which by selection of $n$ equals $\frac{\mathcal{M} (r)}{5\cdot 10p \cdot \log_{1+\eps}(r)}$. Then, \cref{lemma:successful_selection} says that there is a successful selection of $G_1, \ldots, G_p$, which in turn, by \cref{lemma:analysis_if_successful}, means that \cref{alg:chainreduction} succeeds with probability at least $2/3$. As aforementioned, the  communication complexity of \cref{alg:chainreduction} is at most the memory of $\algsub$ plus $p\log(n)$. This contradicts \cref{thm:pindex_hardness} because $n/\log_{1+\epsilon}(r) + p \log(n) < n/(36p^2)$. It follows that $\algsub$ must use memory at least $n/\log_{1+\epsilon}(r)$.

We complete this section with the proof of \cref{lemma:analysis_if_successful}.
\begin{proof}[Proof of \cref{lemma:analysis_if_successful}]
    We first analyze the success probability of \cref{alg:chainreduction} in the $1$-case. In the $1$-case, the elements of $S = \bigcup_{i=1}^p \{(e, o_i) \mid e\in M_i\}$ are elements of the stream, where $M_1, \ldots, M_p$ denote maximum matchings in $G_1, \ldots, G_p$, respectively. Hence, by \cref{property:yescase} of \cref{lemma:gen_hardness_F_properties}, there is a solution to {\MSMBM}  of value at least $p/(1+\epsilon)$. Now, by assumption, $\algsub$ finds a $\frac{p}{1+\varepsilon} \left(\frac{\alpha+1}{p+\alpha+1}  \right) $-approximate solution with probability at least $2/3$. As 
    \[
    \frac{p}{1+\varepsilon} \cdot  \frac{1+\varepsilon}{p} \left( \frac{p+\alpha+1}{1+\alpha} \right)= \frac{p+\alpha+1}{1+\alpha} =1+\frac{1}{1+\alpha} p \enspace,
    \]
    Player $p+1$ correctly outputs $1$-case, i.e., \cref{alg:chainreduction} succeeds, with probability at least $2/3$. 
    
    For the $0$-case, there is no elements  $(e, o_i) \in \cN_i$ in the stream  for $i\in [p]$. Moreover, since the selection of $G_1, \ldots, G_p$ is successful, we have, for a random $0$-case instance from $D(p+1, n)$, that the solution $S$ output by $\algsub$ in \cref{alg:chainreduction} satisfies $|S \cap \cN_i| \leq  m_i / \alpha$ for $i\in [p]$ with probability $9/10\geq 2/3$. Whenever that happens,  \cref{property:nocase} of \cref{lemma:gen_hardness_F_properties} says that $f_{o_1, \ldots, o_p}(S)< 1 + \frac{1}{1+\alpha} p$. It follows that Player $p+1$ outputs $0$-case, i.e., the protocol succeeds, with probability at least $2/3$ for a randomly chosen $0$-case instance from $D(p+1, n)$. Combining the two cases, we have thus proved that the protocol succeeds with probability at least $2/3$ on the distribution $D(p+1, n)$.
\end{proof}

\subsubsection{The selection of \texorpdfstring{$G_1,\ldots, G_p$}{graphs}}
\label{sec:graph_selection}

We now prove \cref{lemma:successful_selection}, which we restate here for convenience.
\successfullemma*

\begin{proof}
We select the graphs $G_1, \ldots, G_p$ one-by-one, starting from the left. When selecting $G_i$, we make sure to select a graph such that, on a random $0$-case instance from $D(p+1, n)$, the probability that \algsub{} in \cref{alg:chainreduction} outputs a set $S$ such that $|S\cap \cN_i| > m_i / \alpha$ is at most $1/(10p)$. The lemma then follows by the union bound.

Now, suppose that we have already selected $G_1, \ldots, G_{i-1}$. We proceed to explain how $G_i$ is selected.  The outline of the argument is as follows. We will simulate the execution of Player $i$ to obtain a streaming algorithm for the bipartite matching problem that uses memory at most $5n$. Hence, by  the selection of $n$ and \cref{corollary:matching-hard-nice}, there must be a graph $G_i$ for which it is likely to fail.

In order to simulate the execution of the $i$-th player, we need to be able to evaluate the oracle calls to the submodular function. To this end,  observe that the value of a submodular function in $\cF$ on a subset  $S \subseteq \cN_{\leq i}$ only depends on the values $m_1, \ldots, m_i$ and the numbers $|(e,j) \in S\cap \cN_k|$ for $j\in [n]$ and $k=1,2,\ldots, i$. 
Since we have selected $G_1, \ldots, G_{i-1}$, we know the values $m_1, \ldots, m_{i-1}$. Furthermore, there are $R$ many possibilities of $m_i$, where $R \leq  1 + \log_{1+\eps}(r)$, because we consider $r$-vertex graphs.  Denote these possibilities by $m_i^1, m_i^2, \ldots, m_i^R$. In the algorithm below that simulates the execution of Player $i$, we make a copy of $\algsub$ for each of these possibilities.  This allows us to answer any evaluations of the submodular function made by Player $i$. 

More precisely, we simulate the execution of Player $i$  to obtain the following algorithm $\algmat$ for the bipartite matching problem in the streaming model:
\begin{description}
\item[Preprocessing] In the unbounded preprocessing phase, we start by sampling a $0$-case instance of \chainPpn  $x^1, x^2, t^2, \ldots, x^{p}, t^{p}, t^{p+1}$ from $D(p+1,n)$. We then simulate the execution of the first $i-1$ players on this instance, which is possible since we fixed the graphs $G_1, \ldots, G_{i-1}$. Finally, we use the state of $\algsub$ received from the previous player (or the initial state if $i = 1$) to initialize   $R$ copies of $\algsub$: $\algsub_1$ for guess $m_i^1$, $\algsub_2$ for guess $m_i^2, \ldots, \algsub_R$ for guess $m_i^R$. 
\item[Processing stream:] At the arrival of an edge $e$, we forward the elements $\{(e,  j) \mid x^i_j = 1\} \subseteq \cN_i$ to each copy $\algsub_j$ of $\algsub$. Note that since each copy has a fixed guess of  $m_i$, we can evaluate any call to the function $f_{o_1,\ldots, o_p}$, where $o_i = t^{i+1}$ for $i=1, \ldots, p-1$.
\item[Postprocessing] In the unbounded postprocessing phase, we decode the memory state  of each copy of $\algsub$ as follows.  Let $Y_j$ be the memory state of $\algsub_j$. Consider all possible streams of edges that could lead to this memory state $Y_j$, let $E(Y_j)$ be the edges that appear in \emph{all} these graphs, and select $M_j$ to be the largest matching in $E(Y_j)$. The output matching is then the largest matching among $M_1, \ldots, M_R$. 
\end{description}

We now bound the memory used by \algmat{} when it processes the stream. It saves the vector $x^i$, which requires $n$ bits, and the indices $o_1, \ldots, o_{p-1}$, which require $p \log(n) < n$ bits. It then runs $R$ parallel copies of $\algsub$, which requires $n/\log_{1+\eps}(r) \cdot R \leq 2n$ bits. Finally, by \cref{property:low_memory} of \cref{lemma:gen_hardness_F_properties}, any submodular function evaluation requires at most $n$ bits. Hence, the total memory usage of the algorithm is at most $5n$. 
Now by the selection of $n$, namely that $5n \leq \mathcal{M}(r)/q$ for $q = 10p$, we can apply  \cref{corollary:matching-hard-nice} to obtain that there is an $r$-vertex graph  such that the described algorithm outputs an $\alpha$-approximate  matching  with probability at most $1/(10p)$. 
We select $G_i$ to be this graph, and let $m_i$ be the smallest power of $1+\eps$ that  upper bounds the size of the maximum matching of $G_i$.

Let $j$ be the guess such that $m_i = m^j_i$. Note that (since $m^j_i= m_i$) $\algsub_j$ in $\algmat$ simulates exactly the distribution of  Player $i$ when given a random $0$-case instance from $D(p+1, n)$. In particular, Player $i$ will send the state $Y_j$ to Player $i+1$, and with probability at least $1-1/(10p)$ (for a random $0$-case instance from $D(p+1, n)$) the size of the largest matching in $E(Y_j)$ is less than  $ m_i / \alpha$. Whenever this happens, the output $S$ of \algsub{} in \cref{alg:chainreduction} must satisfy $|S \cap \cN_i| \leq m_i / \alpha$ (no matter what the following players do) because the edges of $E(Y_j)$ are the only edges that {\algsub} can know for sure (given his memory state when Player $i$ stops executing) that they belong to $G_i$. This completes the description of how the graph $G_i$ is selected. Repeating this until all $p$ graphs have been selected yields the lemma.  
\end{proof}

\bibliographystyle{plain}
\bibliography{ref}

\appendix

\section{Dropping the Assumptions of Theorem~\ref{th:matroid_intersection_alg_simplified}} \label{apx:matroid_intersection_alg}

In \cref{sec:multiple_matroids} we have proved \cref{th:matroid_intersection_alg_simplified}. This theorem makes two simplifying assumptions, namely, that we have an estimate $\tau$ of $f(OPT)$ and that we know $|OPT|$. In the current section we explain how these assumptions can be dropped at the cost of a minor increase in the space complexities guaranteed by the theorem, which proves \cref{th:matroid_intersection_alg_short}.

\subsection{Dropping the assumption of knowing \texorpdfstring{$|OPT|$}{|OPT|}} \label{ssc:no_OPT_size}

In this section we show that the assumption in \cref{th:matroid_intersection_alg_simplified} that $|OPT|$ is known can be dropped without modifying anything else about the theorem. We recall that \cref{th:matroid_intersection_alg_simplified} is proved in \cref{sec:multiple_matroids} using \cref{alg:monotone_SMkM,alg:non-monotone_SMkM}. Both algorithms use the knowledge of $|OPT|$ only to set the number of iterations that they perform (which is set exactly to $|OPT|$). Therefore, a natural way in which one can try to drop the assumption that we know $|OPT|$ is by making the algorithms do $\rho$ iterations, which is always an upper bound on $|OPT|$, and then outputting the best solution produced by these algorithms after any number of iteration.

On the positive side, one can observe that the proofs of \cref{obs:monotone_space_complexity,obs:non-monotone_space_complexity} are unaffected by the suggested modification of the algorithms, and thus, the space complexities guaranteed by \cref{th:matroid_intersection_alg_simplified} apply to the modified algorithms as well. Nevertheless, there are two issues with the suggested modification that need to be addressed.
\begin{itemize}
	\item Consider iteration number $j$ of {\renewcommand{\crefpairconjunction}{{ or }}\cref{alg:monotone_SMkM,alg:non-monotone_SMkM}} for some $j > |OPT|$. In this iteration it is possible that the instances of {\FEA} employed by \cref{alg:monotone_SMkM,alg:non-monotone_SMkM} might return no elements, and therefore, these algorithms might not be able to designate any element as $u_j$.\footnote{Technically, we have the same issue also when an instance of {\FEA} returns only elements that are self loops with respect to at least one of the matroid constraints that it gets. We ignore this possibility because (i) the algorithm {\FEA} designed in \cref{sec:first_element_SMkM} does not ever return self-loops, and (ii) \cref{alg:monotone_SMkM,alg:non-monotone_SMkM} can safely ignore any self loops returned by {\FEA}.} In other words, \cref{alg:monotone_SMkM,alg:non-monotone_SMkM} might get ``stuck'' in an iteration if this iteration is not one of the first $|OPT|$ iterations. To handle this issue, we simply stop executing any thread that gets stuck in this way (we recall that the implementation of \cref{alg:monotone_SMkM,alg:non-monotone_SMkM} that is described in \cref{obs:monotone_space_complexity,obs:non-monotone_space_complexity} is based on multiple parallel threads).
	\item The second issue that we need to handle is that $\rho$ is also unknown, and therefore, we cannot explicitly execute \cref{alg:monotone_SMkM,alg:non-monotone_SMkM} for $\rho$ iterations. To solve this issue, we use the greedy algorithm to track the rank of the set of all the elements that have arrived so far with respect to every one of the matroid constraints. Let us denote this rank with respect to $M_i$ by $r_i$. We also observe that the analysis from \cref{sec:multiple_matroids} works even if we delay some iterations of \cref{alg:monotone_SMkM,alg:non-monotone_SMkM}, as long as, for every integer $1 \leq j \leq |OPT|$, iteration $j$ of these algorithms starts before $o_j$ arrives. This means that it is safe to delay iteration number $j$ of \cref{alg:monotone_SMkM,alg:non-monotone_SMkM} until $\min_{1 \leq i \leq k} r_i \geq j$ because once $o_j$ arrives, we are guaranteed that the set of all elements that have arrived so far includes the set $\{o_1, o_2, \dotsc, o_j\}$ (which is a common independent set of size $j$). Since $\min_{1 \leq j \leq k} r_j$ can never exceed $\rho$, delaying iterations in this way guarantees that \cref{alg:monotone_SMkM,alg:non-monotone_SMkM} will never begin iteration $\rho + 1$, and thus, the number of iterations that they perform is implicitly upper bounded by $\rho$.
\end{itemize}

\subsection{Dropping the assumption of knowing \texorpdfstring{$\tau$}{tau}}

In this section we show how one can drop the assumption in \cref{th:matroid_intersection_alg_simplified} that the algorithms whose existence is guaranteed by the theorem have access to a value $\tau \in [f(OPT), 3f(OPT)]$. We do this by proving the following general proposition, which is strongly based on ideas from~\cite{badanidiyuru2014streaming,haba2020streaming}. One can observe that \cref{th:matroid_intersection_alg_short} follows by plugging into this proposition the modified versions of \cref{alg:monotone_SMkM,alg:non-monotone_SMkM} described in \cref{ssc:no_OPT_size}.

\begin{proposition} \label{prop:no_tau_reduction}
Let $A$ be an arbitrary data stream algorithm for {\SMkM} that assumes access to a value $\tau \in [f(OPT), 3f(OPT)]$ and has an approximation ratio of $r \geq 1$ and a space complexity of $D$ that is independent of $\tau$ (moreover, we need $A$ to have a space complexity of $D$ even when $\tau$ does not belong to above range). 
Then, for every $\eps \in (0, 1/6)$, there exists a data stream algorithm for {\SMkM} that does not assume access to $\tau$ and has an approximation guarantee of $(1 + O(\eps))r$ and a space complexity of $O(\rho \log n + D(\log k + \log \rho + \log \eps^{-1}))$.
\end{proposition}


We prove \cref{prop:no_tau_reduction} using \cref{alg:no_tau_alg}. The main idea of this algorithm is to keep multiple copies of algorithm $A$ in a set $I$, each associated with a different value of $\tau$. At each point of time the algorithm makes sure that either (i) it already has an instance of $A$ associated with a value of $\tau$ close to $f(OPT)$, or (ii) the fraction of the total value of $OPT$ represented by elements of $OPT$ that have already arrived is low. This guarantees that an instance of $OPT$ associated with a value of $\tau$ that is close to $f(OPT)$ is eventually created, and moreover, this instance receives a subset of $OPT$ that carries almost the entire value of $OPT$ itself. For simplicity, we make two assumptions in \cref{alg:no_tau_alg}. The first assumption is that $OPT \neq \varnothing$. If this assumption is violated, then it is trivial to find the optimal solution. The second assumption we make is that the matroids $M_1, M_2, \dotsc, M_k$ have no self-loops. If this assumption does not hold, one can simply discard self-loops upon arrival before the algorithm processes them.
\begin{algorithm}
\caption{Algorithm for Simulating Knowledge of $\tau$} \label{alg:no_tau_alg}
\begin{algorithmic}[1]
\State Let $m \gets 0$, $G \gets \varnothing$ and $I \gets \varnothing$.
\For{every element $u$ that arrives}
	\If{$G + u$ is a common independent set}
		Add $u$ to $G$.
	\EndIf
	\State Update $m \gets \max\{m, f(\{u\})\}$.
	\State \parbox[t]{150mm}{Update $I$ so that it includes an instance of $A$ for every $\tau \in \{2^i \mid i \in \bZ \text{ and } m \leq 2^i \leq 2m(k|G|)^2/\eps\}$. Specifically, if there did not use to be in $I$ an instance of $A$ with $\tau = 2^i$ for some value $i$, and now we need such an instance to exist, then we initialize a new instance of $A$ with this value of $\tau$ and add it to $I$. Similarly, if there used to be in $I$ an instance of $A$ with $\tau = 2^i$ for some value $i$, and now we do not need such an instance to exist anymore, then we remove it from $I$ and delete it from the memory of the algorithm.\strut} \label{line:I_update}
	\State Forward the element $u$ to all the instances of $A$ that currently exist in $I$.
\EndFor
\State \Return{the best among the outputs of all the instances of $A$ in $I$}.
\end{algorithmic}
\end{algorithm}

We begin the analysis of \cref{alg:no_tau_alg} by bounding its space complexity.
\begin{observation}
\Cref{alg:no_tau_alg} requires $O(\rho \log n + D(\log k + \log \rho + \log \eps^{-1}))$ space.
\end{observation}
\begin{proof}
The variables $G$, $u$ and $m$ require $O(\rho \log n)$ space together since $G$ is a common independent set, and therefore, contains at most $\rho$ elements. In addition to these variables, \cref{alg:no_tau_alg} has to store only the instances of $A$ that appear in the set $I$. Each one of these instances requires a space of $D$, and the number of such instances is either $0$ if $m = 0$ or is upper bounded by
\[
	1 + \log_2 \frac{2m(k|G|)^2/\eps}{m}
	=
	1 + \log_2 \frac{2k^2|G|^2}{\eps}
	=
	O(\log k + \log \rho + \log \eps^{-1})
	\enspace,
\]
where the last equality holds since $|G| \leq \rho$ as is argued above. Combining all the above observations, we get that the space complexity of \cref{alg:no_tau_alg} can be upper bounded by
\[
	O(\rho \log n) + D \cdot O(\log k + \log \rho + \log \eps^{-1})
	=
	O(\rho \log n + D(\log k + \log \rho + \log \eps^{-1}))
	\enspace.
	\qedhere
\]
\end{proof}

Let us now denote by $\bar{\tau}$ some power of $2$ within the range $[f(OPT), 2f(OPT)]$. Our plan is to consider the instance of $A$ that is associated with $\tau = \bar{\tau}$ and show two things about this instance. First (\cref{lem:instance_exists}), that this instance appears in $I$ when \cref{alg:no_tau_alg} terminates; and second (\cref{lem:overline_OPT_value}), that this instance receives a subset of $OPT$ that has almost all the value of $OPT$, and therefore, is guaranteed to output a good solution.
\begin{lemma} \label{lem:instance_exists}
An instance of $A$ with $\tau = \bar{\tau}$ exists when \cref{alg:no_tau_alg} terminates.
\end{lemma}
\begin{proof}
The variables $G$ and $m$ determine the set of instances of $A$ that are stored in $I$. Therefore, we begin this proof by providing some bounds on the final values of these variables. One can observe that the set $G$ is constructed using the standard greedy algorithm, which is known to achieve $k$-approximation for the problem of finding a maximum size common independent set. Hence, $|G| \geq |OPT|/k$. Let us now consider the variable $m$, whose final value is $\max_{u \in \cN} f(\{u\})$. By the submodularity and non-negativity of $f$,
\begin{align*}
	m
	={} &
	\max_{u \in \cN} f(\{u\})
	\geq
	|OPT|^{-1} \cdot \sum_{u \in OPT} f(\{u\})
	=
	f(\varnothing) + |OPT|^{-1} \cdot \sum_{u \in OPT} f(u \mid \varnothing)\\
	\geq{} &
	f(\varnothing) + |OPT|^{-1} \cdot f(OPT \mid \varnothing)
	\geq
	|OPT|^{-1} \cdot f(OPT)
	\enspace.
\end{align*}
Additionally, since our assumption that no element is a self-loop implies that $\{u\}$ is a feasible solution for every element $u \in \cN$, we get $m = \max_{u \in \cN} f(\{u\}) \leq f(OPT)$.

Using the above bounds and our assumption that $OPT \neq \varnothing$, we can also get
\[
	\frac{2m(k|G|)^2}{\eps}
	\geq
	\frac{2 \cdot \frac{f(OPT)}{|OPT|} \cdot \left(k \cdot \frac{|OPT|}{k}\right)^2}{\eps}
	=
	\frac{2|OPT| \cdot f(OPT)}{\eps}
	\geq
	2 \cdot f(OPT)
	\enspace,
\]
which implies that \cref{line:I_update} adds an instance of $A$ with $\tau = \bar{\tau}$ at the last iteration of \cref{alg:no_tau_alg} if such an instance did not already appear in $I$ prior to this iteration.
\end{proof}

Let us denote by $A_{\bar{\tau}}$ the instance of $A$ that has $\tau = \bar{\tau}$ and belongs to $I$ according to \cref{lem:instance_exists} when \cref{alg:no_tau_alg} terminates. Additionally, let $\overline{OPT}$ be the optimal solution for the input that is fed to $A_{\bar{\tau}}$ by \cref{alg:no_tau_alg}.
\begin{lemma} \label{lem:overline_OPT_value}
The value $f(\overline{OPT})$ is within the range $[(1 - 2\eps) \cdot f(OPT), f(OPT)]$.
\end{lemma}
\begin{proof}
We note that $A_{\bar{\tau}}$ gets as its input a subset of the full input. Therefore, any feasible solution within the input of $A_{\bar{\tau}}$ is a feasible solution also with respect to the full input, and thus, has a value of at most $f(OPT)$. Accordingly, we concentrate in the rest of this proof on showing that there is a feasible solution for the input instance received by $A_{\bar{\tau}}$ of value at least $(1 - 2\eps) \cdot f(OPT)$.

Let us denote by $o_1, o_2, \dotsc, o_h$ the elements of $OPT$ that arrived before the instance $A_{\bar{\tau}}$ was created. For every integer $1 \leq i \leq h$, we must have that either an instance of $A$ associated with $\tau = \bar{\tau}$ did not exist at the time in which $o_i$ arrived, or that the instance of this kind that existed when $o_i$ arrived was removed at a later point (because $A_{\bar{\tau}}$ is an instance associated with this value of $\tau$ that was created later). This means that there was some time $t$ in the execution of \cref{alg:no_tau_alg}, which was either immediately following the arrival of $o_i$ or a later time point, in which no instance of $A$ associated with $\tau = \bar{\tau}$ existed. If we denote by $G_t$ and $m_t$ the values of $G$ and $m$ at the time $t$, then the non-existence of an instance of $A$ associated with $\tau = \bar{\tau}$ at this time implies that at least one of the following inequalities is violated.
\[
	m_t \leq \bar{\tau} \leq \frac{2m_t(k|G_t|)^2}{\eps}
	\enspace.
\]
However, the left inequality is guaranteed to hold because $m_t \leq \max_{u \in \cN} f(\{u\}) \leq f(OPT) \leq \bar{\tau}$. Hence, it must be the case that the right inequality is violated. In other words, we have
\begin{equation} \label{eq:discarded_element_value}
	\bar{\tau}
	>
	\frac{2m_t(k|G_t|)^2}{\eps}
	\geq
	\frac{2m_t i^2}{\eps}
	\enspace,
\end{equation}
where the second inequality holds because $|G|$ is a $k$-approximation for the maximum common independent set size, and following the arrival of $o_i$ we are guaranteed that the common independent set $\{o_1, o_2, \dotsc, o_i\}$ of size $i$ already arrived. 

Since $m_t$ is the maximum value of any singleton set consisting of an element that arrived before time $t$, it upper bounds $f(\{o_i\})$. Therefore, Inequality~\eqref{eq:discarded_element_value} implies
\[
	f(\{o_i\})
	\leq
	m_t
	<
	\frac{\eps \bar{\tau}}{2 i^2}
	\leq
	\frac{\eps}{i^2} \cdot f(OPT)
	\enspace.
\]
Consider now the set $OPT' = OPT \setminus \{o_1, o_2, \dotsc, o_h\}$. By the definition of the elements $o_1, o_2, \dotsc, o_h$, all the elements of $OPT'$ are fed into the instance $A_{\bar{\tau}}$ of $A$. Furthermore, $OPT'$ is a common independent set because it is a subset of the common independent set $OPT$, and its value is
\begin{align*}
	f(OPT')
	\geq{} &
	f(OPT) - \sum_{i = 1}^h f(o_i \mid \varnothing)
	\geq
	f(OPT) - \sum_{i = 1}^h f(o_i)\\
	\geq{} &
	f(OPT) - \sum_{i = 1}^h \frac{\eps}{i^2} \cdot f(OPT)
	\geq
	\left[1 - \eps - \int_1^\infty \frac{\eps dx}{x^2}\right] \cdot f(OPT)
	=
	(1 - 2\eps) \cdot f(OPT)
	\enspace,
\end{align*}
where the first inequality follows from the submodularity of $f$ and the second inequality follows from its non-negativity.
\end{proof}

We are now ready to prove \cref{prop:no_tau_reduction}.
\begin{proof}[Proof of \cref{prop:no_tau_reduction}]
Consider the instance $A_{\bar{\tau}}$ of $A$ defined above. According to the definition of $\bar{\tau}$,
\[
	\bar{\tau}
	\in
	[f(OPT), 2f(OPT)]
	\subseteq
	\left[f(\overline{OPT}), \frac{2f(\overline{OPT})}{1 - 2\eps}\right]
	\subseteq
	[f(\overline{OPT}), 3f(\overline{OPT})]
	\enspace,
\]
where the first inclusion holds by \cref{lem:overline_OPT_value}, and the second inclusion holds since $\eps \leq 1/6$. We now observe that the membership of $\bar{\tau}$ within this range implies that, by the approximation guarantee of $A$, the instance $A_{\bar{\tau}}$ is guaranteed to output a solution of value at least
\[
	\frac{1}{r}\cdot f(\overline{OPT})
	\geq
	\frac{1 - 2\eps}{r} \cdot f(OPT)
	\geq
	\frac{1}{(1 + 3\eps)r} \cdot f(OPT)
	\enspace,
\]
where the last inequality holds again for $\eps \in (0, 1/6)$.
\end{proof}

\end{document}